\documentclass[twocolumn]{autart}  
\usepackage{cite}
\usepackage{amsmath,amssymb,amsfonts}
\usepackage{algorithm}
\usepackage{algpseudocode}
\usepackage{graphicx}

\usepackage{textcomp}
\usepackage{xcolor}
\usepackage{soul,dsfont} 
\usepackage{tikz} 
\usepackage{color}
\usepackage{algorithm} 
\usepackage{array}
\usepackage{eqparbox}
\usepackage{url}
\usepackage{stfloats}
\usepackage{graphics}
\usepackage[mathscr]{euscript}
\graphicspath{{figures/}{./}}

\newtheorem{proposition}{Proposition}

\newtheorem{remark}{Remark}

\begin{document}

\begin{frontmatter}

\title{Data-Driven Tensor Decomposition Identification of Homogeneous Polynomial Dynamical Systems\thanksref{footnoteinfo}}

\thanks[footnoteinfo]{This paper was not presented at any IFAC 
meeting. Corresponding author Can Chen.}

\author[1]{Xin Mao}\ead{xinm@unc.edu},
\author[2]{Joshua Pickard}\ead{jpickard@broadinstitute.org},
\author[1,3,4]{Can Chen}\ead{canc@unc.edu} 

\address[1]{School of Data Science and Society, University of North Carolina at Chapel Hill, Chapel Hill, NC 27599, USA}  
\address[2]{Eric and Wendy Schmidt Center, Broad Institute of MIT and Harvard, Cambridge, MA 02142, USA}
\address[3]{Department of Mathematics, University of North Carolina at Chapel Hill, Chapel Hill, NC 27599, USA}  
\address[4]{Department of Biostatistics, University of North Carolina at Chapel Hill, Chapel Hill, NC 27599, USA}  


\begin{abstract}
Homogeneous polynomial dynamical systems (HPDSs), which can be equivalently represented by tensors, are essential for modeling  higher-order networked systems, including ecological networks, chemical reactions, and multi-agent robotic systems. However, identifying such systems from data is challenging due to the rapid growth in the number of parameters with increasing system dimension and polynomial degree. In this article, we adopt compact and scalable representations of HPDSs leveraging low-rank tensor decompositions, including  tensor train, hierarchical Tucker, and  canonical polyadic decompositions. These representations exploit the intrinsic multilinear structure of HPDSs and substantially reduce the dimensionality of the parameter space. Rather than identifying the full dynamic tensor, we develop a data-driven framework that directly learns the underlying factor tensors or matrices in the associated decompositions from time-series data. The resulting identification problem is solved using alternating least-squares algorithms tailored to each tensor decomposition, achieving both accuracy and computational efficiency. We further analyze the robustness of the proposed framework in the presence of measurement noise and characterize data informativity. Finally, we demonstrate the effectiveness of our framework with numerical examples. 
	\end{abstract}
    \end{frontmatter}

    \section{Introduction}
Homogeneous polynomial dynamical systems (HPDSs) play a central role in modeling networked systems with fixed-order higher-order interactions, such as ecological networks with cubic couplings and chemical reaction systems governed by homogeneous rate laws, where all terms share the same total degree with respect to concentrations \cite{bairey2016high,chellaboina2009modeling,grilli2017higher,cencetti2021temporal,craciun2019polynomial,mao2025model,chesi2009homogeneous,samardzija1983stability}. Beyond these classical examples, any polynomial dynamical system can be homogenized, allowing HPDSs to represent higher-order interactions of varying orders and extending their applicability to diverse complex systems, including fluid dynamics, epidemiology, and robotics \cite{lassila2014model,carlberg2013gnat,brauer2019mathematical,guckenheimer2013nonlinear,malizia2024reconstructing,cox2020applications,ivakhnenko2007polynomial,najm2009uncertainty}. Accurate modeling of such systems requires identifying the underlying HPDSs from time-series data. This task is essential for predicting system behavior, analyzing system-theoretic properties, uncovering hidden nonlinear interactions, and informing the design of effective control strategies. Consequently, the system identification of HPDSs constitutes a fundamental step toward leveraging these models in both theoretical studies and practical applications.

A variety of system identification techniques have been proposed for linear and nonlinear dynamical systems. Classical linear methods, such as subspace identification \cite{van2012subspace, jansson1998consistency}, ARX models \cite{isaksson2002identification,jansson2003subspace}, and least-squares regression \cite{ljung1998system,watson1967linear}, have proven effective for moderate-scale systems, offering computational efficiency and reliable parameter estimates. To capture richer nonlinear behaviors, nonlinear ARX models \cite{zhao2010recursive,de1997stabilizing}, Volterra series expansions \cite{batselier2017tensor,brockett1976volterra}, sparse regression techniques \cite{bertsimas2020sparse,bertsimas2020sparsehigh}, polynomial approximations \cite{ostertagova2012modelling,guo2020constructing}, and SINDy \cite{brunton2016discovering,delabays2025hypergraph} have been applied to identify governing equations from data by exploiting structural or sparsity assumptions. More recently, data-driven and machine learning-based methods, such as dynamic mode decomposition and physics-informed network models, have demonstrated the ability to approximate high-dimensional nonlinear dynamics directly from observations \cite{mao2025identification,proctor2016dynamic,pfluger2010spatially}. However, applying these methods to HPDSs faces significant challenges due to the combinatorial growth of parameters with increasing state dimension and polynomial degree, which leads to computational intractability and overfitting in high-dimensional settings.

In recent years, numerous studies have explored tensor representations for polynomial dynamical systems, recognizing that the coefficient structure of an HPDS can be naturally expressed as a high-order tensor, often called the dynamic tensor in analogy to the dynamic matrix of linear dynamical systems \cite{batselier2017tensor,dong2024controllability,batselier2022low,chen2022explicit,chen2021controllability,gorodetsky2018high,pickard2023observability,pickard2024geometric,pickard2024kronecker,cui2024discrete,cui2025analysis}. Such tensor-based representations provide an intuitive and expressive framework for encoding the multilinear interactions inherent in HPDSs and enable the use of tools from tensor algebra for system identification and analysis. Moreover, these formulations make it possible to exploit structural properties such as symmetry and sparsity that are difficult to capture using conventional approaches. However, directly estimating the full dynamic tensor of an HPDS from data is computationally prohibitive, as the number of tensor entries grows combinatorially with the system dimension and polynomial degree. Consequently, existing tensor-based identification methods that rely on full tensor representations are limited by the curse of dimensionality and often fail to scale to high-dimensional systems or higher-order nonlinearities.

To overcome this limitation, a natural approach is to leverage low-rank tensor decompositions, which can drastically reduce both computational and storage complexity while preserving modeling expressiveness. Several structured tensor decompositions have been proposed in the literature, including tensor train decomposition (TTD) \cite{oseledets2011tensor,oseledets2009breaking}, hierarchical Tucker decomposition (HTD) \cite{lubich2013dynamical,grasedyck2010hierarchical}, and canonical polyadic decomposition (CPD) \cite{phan2013candecomp,hong2020generalized}. TTD and HTD provide efficient and scalable representations of high-order tensors by decomposing them into a sequence or tree of lower-dimensional factors and are valued for their numerical stability. In contrast, CPD represents a tensor as a sum of rank-one components offering high compactness and is widely used for its conceptual simplicity and strong uniqueness properties under mild conditions. Although low-rank tensor decompositions have been extensively studied for data compression and numerical approximation \cite{bousse2017tensor,batselier2022low,kargas2020nonlinear,phan2020stable,yan2014image}, their application to system identification, particularly for nonlinear dynamical systems, remains relatively underexplored.

In this article, we develop a novel data-driven framework for identifying HPDSs under structured low-rank tensor parameterizations. We consider TTD-, HTD-, and CPD-based representations of HPDSs with unknown dynamic tensors and formulate the identification problem as a structured nonlinear least-squares optimization using time-series data. By exploiting the multilinear structure inherent to each decomposition, we decompose the nonlinear optimization into a sequence of subproblems that reduce to standard linear least-squares problems, whose dimensions depend only on the system size and the tensor ranks. Instead of recovering the full dynamic tensor, our framework directly identifies the low-rank tensor or matrix factors from time-series data, bypassing the need for full tensor reconstruction. As a result, the proposed methods scale efficiently to high-dimensional systems and higher polynomial degrees. We further establish the convergence of the proposed algorithms, analyze their robustness in the presence of  noise, and discuss the data informativity conditions necessary for system identification under each decomposition.

The proposed data-driven tensor decomposition identification framework has broad potential applications in real-world higher-order networked systems, including ecological networks, biological networks, and multi-agent robotic systems. For instance, it enables the recovery of structured multi-species interaction mechanisms directly from population time-series data while avoiding the combinatorial explosion associated with full tensor representations. This approach provides a scalable and interpretable means of uncovering latent ecological interaction patterns that extend beyond simple pairwise effects, with important implications for system analysis and the design of optimal strategies for ecological conservation. The remainder of this article is organized as follows. In Section~\ref{sec: preliminary}, we review key concepts in tensor algebra, including tensor-vector products, tensor matricization, and tensor decompositions. Section~\ref{sec: identification} develops the data-driven  identification framework for TTD-, HTD-, and CPD-based HPDSs, leveraging alternating least-squares methods. Section~\ref{sec: simulation} presents numerical examples that illustrate the effectiveness of the proposed algorithms. Finally, Section~\ref{sec: conclusion} concludes the article and discusses directions for future research.

  \section{Preliminaries}\label{sec: preliminary}
  A tensor is a multidimensional array that generalizes scalars, vectors, and matrices to higher-order structures \cite{kolda2009tensor, chen2024tensor}. The order of a tensor refers to the number of modes it possesses. A $k$th-order tensor is often denoted by $\mathscr A\in\mathbb R^{n_1\times n_2\times\cdots\times n_k}$, where $n_j$ is the size of $j$th mode. A tensor is called cubical if all of its modes have the same size, i.e., $n_1=n_2=\cdots=n_k$. A cubical tensor is said to be almost symmetric if its entries are invariant under any
permutation of the first $k-1$ indices.

\subsection{Tensor Products}
The outer product of two tensors generalizes the vector outer product. Let $\mathscr{A} \in \mathbb{R}^{n_1 \times n_2 \times \cdots \times n_{k_1}}$ and $\mathscr{B} \in \mathbb{R}^{m_1 \times m_2 \times \cdots \times m_{k_2}}$ be two tensors. Their outer product, denoted by $\mathscr{A} \circ \mathscr{B} \in \mathbb{R}^{n_1 \times \cdots \times n_{k_1} \times m_1 \times \cdots \times m_{k_2}}$, is defined entrywise as
\[
(\mathscr A\circ\mathscr B)_{i_1i_2\cdots i_{k_1}j_1j_2\cdots j_{k_2}}=\mathscr A_{i_1i_2\cdots i_{k_1}}\mathscr B_{j_1j_2\cdots j_{k_2}}.
\]

The tensor-vector product is a natural extension of the matrix-vector product. Let $\mathscr{A} \in \mathbb{R}^{n_1 \times n_2 \times \cdots \times n_k}$ be a $k$th-order tensor and $\textbf{v} \in \mathbb{R}^{n_p}$ a vector. The product of $\mathscr{A}$ along mode $p$, denoted by $\mathscr{A} \times_p \textbf{v} \in \mathbb{R}^{n_1 \times \cdots \times n_{p-1} \times n_{p+1} \times \cdots \times n_k}$, is defined entrywise as
\[
(\mathscr A \times_p \textbf v)_{j_1j_2\cdots j_{p-1}j_{p+1}\cdots j_k}=\sum_{j_p=1}^{n_p}\mathscr A_{j_1j_2\cdots j_k}\textbf v_{j_p}.
\]
This operation can be naturally extended to all modes of $\mathscr{A}$, known as the Tucker product
\begin{align*}
\mathscr A\times_1 \textbf v_1 \times_2 \textbf v_2\times_3\dots\times_k\textbf v_k \in\mathbb{R},
\end{align*}
where $\textbf{v}_p\in\mathbb{R}^{n_p}$ for $p=1,2,\dots,k$. If $\textbf v_p=\textbf v$ for all $p$, we write the product as $\mathscr A\textbf v^{k}=\mathscr A\times_1 \textbf v \times_2 \textbf v\times_3\dots\times_k\textbf v$ for simplicity.

The Kronecker product plays a fundamental role in tensor algebra. For matrices
$\textbf A\in\mathbb R^{m\times n}$ and
$\textbf B\in\mathbb R^{p\times q}$, their Kronecker product is defined as
\[
\textbf A\otimes \textbf B
=
\begin{bmatrix}
\textbf{A}_{11}\textbf B & \textbf{A}_{12}\textbf B & \cdots & \textbf{A}_{1n}\textbf B\\
\vdots & \vdots & \ddots & \vdots\\
\textbf{A}_{m1}\textbf B & \textbf{A}_{m2}\textbf B & \cdots & \textbf{A}_{mn}\textbf B
\end{bmatrix}\in\mathbb{R}^{mp\times nq}.
\]
The Kronecker product satisfies several useful properties that will be used throughout this paper:
(i) mixed-product property:
$(\textbf A\otimes\textbf B)(\textbf C\otimes\textbf D)
=
(\textbf A\textbf C)\otimes(\textbf B\textbf D)$,
whenever the products are well-defined; (ii) vectorization identity:
$\mathrm{vec}(\textbf A\textbf X\textbf B)
=
(\textbf B^\top\otimes\textbf A)\mathrm{vec}(\textbf X)$, where $\mathrm{vec}(\cdot)$ denotes the vectorization operation.

The Khatri-Rao product of two matrices
$\textbf A=[\textbf a_1 \ \textbf a_2 \ \cdots \ \textbf a_r]\in\mathbb R^{m\times r}$
and
$\textbf B=[\textbf b_1 \ \textbf b_2 \ \cdots \ \textbf b_r]\in\mathbb R^{n\times r}$
is defined as the column-wise Kronecker product
\[
\textbf A\odot\textbf B
=
[\textbf a_1\otimes\textbf b_1\quad\cdots\quad\textbf a_r\otimes\textbf b_r]\in\mathbb{R}^{mn\times r}.
\]
In particular, for any vectors $\textbf u$, $\textbf v\in\mathbb R^r$,
it holds that $
\mathrm{diag}(\textbf u)\,\textbf v
=
\textbf u\odot \textbf v$,
which follows directly from the column-wise definition of the Khatri-Rao product.

\subsection{Tensor Matricization}
Tensor matricization reshapes a  tensor into a  matrix or a  vector \cite{kolda2009tensor,ragnarsson2012block}. Consider a $k$th-order tensor $\mathscr{A}\in\mathbb{R}^{n_1\times n_2\times\cdots\times n_k}$. We partition its modes into two ordered, disjoint index sets,
$\mathcal{R}=\{r_1,r_2,\dots,r_d\}$ and $
\mathcal{C}=\{c_1,c_2,\dots,c_{k-d}\}$,
which specify the modes assigned to the rows and columns of the the matricization, respectively. To explicitly describe how tensor indices are mapped to matrix indices, we define the following index mapping function
\begin{equation*}
\psi(\{j_1,\dots,j_k\},\{n_1,\dots,n_k\})
= j_1+\sum_{i=2}^k (j_i-1)\prod_{l=1}^{i-1} n_l,
\end{equation*}
which maps a tensor index $\{j_1,j_2,\dots,j_k\}$, with $1\leq j_i\le n_i$, to a unique index following column-major ordering. The row and column indices of the matricized tensor are given by
$r=\psi(\{j_{r_1},\dots,j_{r_d}\},\{n_{r_1},\dots,n_{r_d}\})$ and $c=\psi(\{j_{c_1},\dots,j_{c_{k-d}}\},\{n_{c_1},\dots,n_{c_{k-d}}\})$, respectively. 

The mode-$\mathcal{R}$ matricization of $\mathscr{A}$ is defined as $$\textbf{A}_{(\mathcal{R})}\in\mathbb{R}^{\prod_{p=1}^{d}n_{r_p}\times\prod_{p=1}^{k-d}n_{c_{p}}},$$ whose entries satisfy $(\textbf A_{(\mathcal R)})_{rc} = \mathscr A_{j_1 j_2 \cdots j_k}$. A commonly used special case is the mode-$p$ matricization, obtained by taking
$\mathcal{R}=\{p\}$ and  $\mathcal{C}=\{1,2,\dots,p-1,p+1,\dots,k\}$, which yields the matrix
$\textbf{A}_{(p)}\in\mathbb{R}^{n_p\times(n_1n_2\cdots n_{p-1}n_{p+1}\cdots n_k)}$,
whose rows correspond to all mode-$p$ fibers of the tensor. 
Moreover, choosing all modes as row indices, i.e.,
$\mathcal{R}=\{1,2,\dots,k\}$ and $\mathcal{C}=\emptyset$,
reduces the matricization to the standard vectorization.

\subsection{Tensor Decompositions}

\subsubsection{Tensor Train Decomposition}
Tensor train decomposition (TTD) provides a numerically robust and scalable representation for high-dimensional tensors by decomposing them into a sequence of interconnected third-order factor tensors, see Fig. \ref{fig: ttd}. The resulting sequential structure enables efficient tensor contractions and stable manipulation of high-order tensors in large-scale settings. For a $k$th-order tensor $\mathscr A\in\mathbb{R}^{n_1\times n_2\times\dots\times n_k}$, the TTD form is mathematically expressed as
\begin{equation}\label{eq:ttd}
\mathscr A=\sum_{j_0=1}^{r_0}\sum_{j_1=1}^{r_1}\cdots\sum_{j_k=1}^{r_k}\mathscr T_{j_0:j_1}^{(1)}\circ \mathscr T_{j_1:j_2}^{(2)}\circ\dots\circ \mathscr T_{j_{k-1}:j_k}^{(k)},
\end{equation}
where $\{r_0, r_1,\dots, r_k\}$ are the TT-ranks with $r_0=r_k=1$, and $\mathscr T^{(p)}\in\mathbb{R}^{r_{p-1}\times n_p\times r_p}$ are third-order factor tensors. The colon ``:'' indicates that all indices along the corresponding dimension are
included, analogous to the colon operator in MATLAB. 

TTD representations are not unique due to an intrinsic gauge freedom, meaning that adjacent
factor tensors may be transformed by arbitrary nonsingular matrices without changing the
represented tensor. To obtain a numerically stable and essentially unique representation, canonical
forms such as left-orthogonality or mixed orthogonality are typically imposed, which remove this gauge ambiguity. A notable advantage of TTD is its numerical stability, which ensures that the TT-ranks of a given tensor $\mathscr A$ can be determined in a robust and
computationally reliable manner.

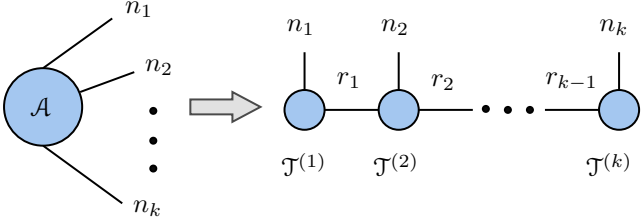
\begin{figure}[t]
\centering
\tikzset{every picture/.style={line width=0.75pt}} 

\begin{tikzpicture}[x=0.75pt,y=0.75pt,yscale=-1,xscale=1]

\draw    (153,129) -- (126,139) ;
\draw    (107.5,166) -- (147,195) ;
\draw    (107.5,127) -- (142,102) ;
\draw  [fill={rgb, 255:red, 74; green, 144; blue, 226 }  ,fill opacity=0.5 ] (88,146.5) .. controls (88,135.73) and (96.73,127) .. (107.5,127) .. controls (118.27,127) and (127,135.73) .. (127,146.5) .. controls (127,157.27) and (118.27,166) .. (107.5,166) .. controls (96.73,166) and (88,157.27) .. (88,146.5) -- cycle ;
\draw    (248,148) -- (275,148) ;
\draw  [color={rgb, 255:red, 74; green, 74; blue, 74 }  ,draw opacity=1 ][fill={rgb, 255:red, 155; green, 155; blue, 155 }  ,fill opacity=0.3 ] (181,142.75) -- (202,142.75) -- (202,139) -- (216,146.5) -- (202,154) -- (202,150.25) -- (181,150.25) -- cycle ;
\draw  [color={rgb, 255:red, 0; green, 0; blue, 0 }  ,draw opacity=1 ][fill={rgb, 255:red, 74; green, 144; blue, 226 }  ,fill opacity=0.5 ] (228,148) .. controls (228,142.48) and (232.48,138) .. (238,138) .. controls (243.52,138) and (248,142.48) .. (248,148) .. controls (248,153.52) and (243.52,158) .. (238,158) .. controls (232.48,158) and (228,153.52) .. (228,148) -- cycle ;
\draw  [fill={rgb, 255:red, 0; green, 0; blue, 0 }  ,fill opacity=0.96 ] (161,148.5) .. controls (161,147.67) and (161.67,147) .. (162.5,147) .. controls (163.33,147) and (164,147.67) .. (164,148.5) .. controls (164,149.33) and (163.33,150) .. (162.5,150) .. controls (161.67,150) and (161,149.33) .. (161,148.5) -- cycle ;
\draw  [fill={rgb, 255:red, 0; green, 0; blue, 0 }  ,fill opacity=0.96 ] (161,163.5) .. controls (161,162.67) and (161.67,162) .. (162.5,162) .. controls (163.33,162) and (164,162.67) .. (164,163.5) .. controls (164,164.33) and (163.33,165) .. (162.5,165) .. controls (161.67,165) and (161,164.33) .. (161,163.5) -- cycle ;
\draw  [fill={rgb, 255:red, 0; green, 0; blue, 0 }  ,fill opacity=0.96 ] (161,177.5) .. controls (161,176.67) and (161.67,176) .. (162.5,176) .. controls (163.33,176) and (164,176.67) .. (164,177.5) .. controls (164,178.33) and (163.33,179) .. (162.5,179) .. controls (161.67,179) and (161,178.33) .. (161,177.5) -- cycle ;
\draw  [fill={rgb, 255:red, 74; green, 144; blue, 226 }  ,fill opacity=0.5 ] (275,148) .. controls (275,142.48) and (279.48,138) .. (285,138) .. controls (290.52,138) and (295,142.48) .. (295,148) .. controls (295,153.52) and (290.52,158) .. (285,158) .. controls (279.48,158) and (275,153.52) .. (275,148) -- cycle ;
\draw    (238,138) -- (238,119) ;
\draw    (285,138) -- (285,119) ;
\draw    (295,148) -- (322,148) ;
\draw    (358,148) -- (385,148) ;
\draw  [fill={rgb, 255:red, 74; green, 144; blue, 226 }  ,fill opacity=0.5 ] (385,148) .. controls (385,142.48) and (389.48,138) .. (395,138) .. controls (400.52,138) and (405,142.48) .. (405,148) .. controls (405,153.52) and (400.52,158) .. (395,158) .. controls (389.48,158) and (385,153.52) .. (385,148) -- cycle ;
\draw    (395,138) -- (395,119) ;
\draw  [fill={rgb, 255:red, 0; green, 0; blue, 0 }  ,fill opacity=0.96 ] (327,147.5) .. controls (327,146.67) and (327.67,146) .. (328.5,146) .. controls (329.33,146) and (330,146.67) .. (330,147.5) .. controls (330,148.33) and (329.33,149) .. (328.5,149) .. controls (327.67,149) and (327,148.33) .. (327,147.5) -- cycle ;
\draw  [fill={rgb, 255:red, 0; green, 0; blue, 0 }  ,fill opacity=0.96 ] (338,148) .. controls (338,147.17) and (338.67,146.5) .. (339.5,146.5) .. controls (340.33,146.5) and (341,147.17) .. (341,148) .. controls (341,148.83) and (340.33,149.5) .. (339.5,149.5) .. controls (338.67,149.5) and (338,148.83) .. (338,148) -- cycle ;
\draw  [fill={rgb, 255:red, 0; green, 0; blue, 0 }  ,fill opacity=0.96 ] (349,148) .. controls (349,147.17) and (349.67,146.5) .. (350.5,146.5) .. controls (351.33,146.5) and (352,147.17) .. (352,148) .. controls (352,148.83) and (351.33,149.5) .. (350.5,149.5) .. controls (349.67,149.5) and (349,148.83) .. (349,148) -- cycle ;

\draw (147,92) node [anchor=north west][inner sep=0.75pt]   [align=left] {$\displaystyle n_{1}$};
\draw (151,192) node [anchor=north west][inner sep=0.75pt]   [align=left] {$\displaystyle n_{k}$};
\draw (145,159) node [anchor=north west][inner sep=0.75pt]   [align=left] {$ $};
\draw (157,121) node [anchor=north west][inner sep=0.75pt]   [align=left] {$\displaystyle n_{2}$};
\draw (227.7,100.84) node [anchor=north west][inner sep=0.75pt]   [align=left] {$\displaystyle n_{1}$};
\draw (383.7,100.96) node [anchor=north west][inner sep=0.75pt]   [align=left] {$\displaystyle n_{k}$};
\draw (274.78,101.32) node [anchor=north west][inner sep=0.75pt]   [align=left] {$\displaystyle n_{2}$};
\draw (252.7,127.84) node [anchor=north west][inner sep=0.75pt]   [align=left] {$\displaystyle r_{1}$};
\draw (299.7,128.84) node [anchor=north west][inner sep=0.75pt]   [align=left] {$\displaystyle r_{2}$};
\draw (356.7,127.84) node [anchor=north west][inner sep=0.75pt]   [align=left] {$\displaystyle r_{k-1}$};
\draw (100,140) node [anchor=north west][inner sep=0.75pt]    {$\mathscr{A}$};
\draw (225,168) node [anchor=north west][inner sep=0.75pt]    {$\mathscr{T}^{(1)}$};
\draw (271,168) node [anchor=north west][inner sep=0.75pt]    {$\mathscr{T}^{(2)}$};
\draw (377,168) node [anchor=north west][inner sep=0.75pt]    {$\mathscr{T}^{(k)}$};
\end{tikzpicture}
\caption{Illustration of the TT decomposition of a $k$th-order tensor into a sequence of third-order factor tensors.}\label{fig: ttd}
\end{figure}

\subsubsection{Hierarchical Tucker Decomposition} 
Hierarchical Tucker decomposition (HTD) provides a tree-structured representation of high-order tensors by recursively grouping tensor modes according to a binary dimension tree.
The binary tree $\mathcal T$ satisfies the following properties:
    (i) each node represents a subset of the tensor modes $\{1,2,\dots,k\}$;
    (ii) the root node corresponds to the full mode set $\{1,2,\dots,k\}$;
    (iii) each leaf node corresponds to a single mode;
    (iv) each parent node is the disjoint union of its two children.
An example of a dimension tree is shown in Fig.~\ref{fig: htdtree}.
The level of a node is defined as its distance from the root, and the depth of the tree is $d=\lceil\log_2 k\rceil$.

For a $k$th-order tensor $\mathscr A\in\mathbb{R}^{n_1\times n_2\times\dots\times n_k}$, each node $\mathcal P\in\mathcal T$ is associated with a factor matrix $\textbf V_\mathcal P$, which spans the column space of the mode-$\mathcal P$ matricization $\textbf A_{(\mathcal P)}$. A key property of HTD is that factor matrices $\textbf V_\mathcal P$ associated with internal nodes do not need to be stored explicitly. Instead, they are constructed
recursively from the factor matrices of the left and right child nodes,
$\textbf V_{\mathcal P_l}$ and $\textbf V_{\mathcal P_r}$, according to 
\begin{equation}\label{eq:HTrelation}
\textbf V_\mathcal P=(\textbf V_{\mathcal P_l}\otimes \textbf V_{\mathcal P_r})
\textbf C_\mathcal P,
\end{equation}
where $\textbf C_\mathcal P\in\mathbb{R}^{r_{\mathcal P_l}r_{\mathcal P_r}\times r_\mathcal P}$ is the transfer matrix, and $r_{\mathcal P_l}$, $r_{\mathcal P_r}$, and $r_\mathcal P$ denote the hierarchical ranks associated with the nodes $\mathcal P_l$, $\mathcal P_r$, and $\mathcal P$, respectively. Consequently, an HT representation is fully characterized by the factor matrices at the leaf nodes $\textbf V_p\in\mathbb{R}^{n_p\times r_p}$ and the transfer matrices
$\textbf C_\mathcal P$ at internal nodes. 

HTD can be viewed as a generalization of TTD that allows more flexible groupings of tensor modes. As with TTD, the HTD representation is not unique due to gauge freedom. For convenience and notational uniformity, we associate a transfer matrix with every node in the dimension tree, and set $\textbf C_\mathcal P=\textbf I_{r_\mathcal P}$
whenever $\mathcal P$ is a leaf node. Applying  (\ref{eq:HTrelation}) recursively from the leaves to the root (whose rank is set to $1$) yields the vectorized form
\begin{align*}
\scalebox{0.95}{
$\text{vec}({\mathscr A})\!=\!(\textbf V_{k}\otimes\cdots\otimes\textbf V_{1})(\otimes_{\mathcal Q\in\mathcal G_{d-1}}\textbf C_\mathcal Q)\cdots(\otimes_{\mathcal Q\in\mathcal G_{0}}\textbf C_\mathcal Q),$}
\end{align*}
where $\mathcal G_j$ denote the set of nodes at level $j$ of the dimension tree for $j=0,1,\dots,d-1$. This formulation is valid for any (possibly non-complete) balanced dimension tree, since leaf nodes that appear above the bottom level simply contribute identity matrices in the Kronecker products.

\begin{figure}[t]
\centering
\tikzset{every picture/.style={line width=0.75pt}} 

\begin{tikzpicture}[x=0.75pt,y=0.75pt,yscale=-1,xscale=1]

\draw [color={rgb, 255:red, 74; green, 144; blue, 226 }  ,draw opacity=1 ][line width=2.25]    (330,65.5) -- (279,83.5) ;
\draw [color={rgb, 255:red, 74; green, 144; blue, 226 }  ,draw opacity=1 ][line width=2.25]    (329,65.5) -- (380,83) ;
\draw [color={rgb, 255:red, 74; green, 144; blue, 226 }  ,draw opacity=1 ][line width=2.25]    (276,107.5) -- (241,127.5) ;
\draw [color={rgb, 255:red, 74; green, 144; blue, 226 }  ,draw opacity=1 ][line width=2.25]    (312,128) -- (275,106.5) ;
\draw [color={rgb, 255:red, 74; green, 144; blue, 226 }  ,draw opacity=1 ][line width=2.25]    (380,107.5) -- (351,128.5) ;
\draw [color={rgb, 255:red, 74; green, 144; blue, 226 }  ,draw opacity=1 ][line width=2.25]    (421,128) -- (379,107.5) ;
\draw [color={rgb, 255:red, 74; green, 144; blue, 226 }  ,draw opacity=1 ][line width=2.25]    (420,151.5) -- (389,170.5) ;
\draw [color={rgb, 255:red, 74; green, 144; blue, 226 }  ,draw opacity=1 ][line width=2.25]    (454,170.5) -- (419,151.5) ;
\draw   (288,40.5) -- (372,40.5) -- (372,63.5) -- (288,63.5) -- cycle ;
\draw   (258,83.5) -- (297,83.5) -- (297,105.5) -- (258,105.5) -- cycle ;
\draw   (352,83.5) -- (406,83.5) -- (406,105.5) -- (352,105.5) -- cycle ;
\draw   (227,127.5) -- (251,127.5) -- (251,149.5) -- (227,149.5) -- cycle ;
\draw   (336,128.5) -- (360,128.5) -- (360,150.5) -- (336,150.5) -- cycle ;
\draw   (298,128.5) -- (322,128.5) -- (322,150.5) -- (298,150.5) -- cycle ;
\draw   (399,128.5) -- (437,128.5) -- (437,150.5) -- (399,150.5) -- cycle ;
\draw   (377,170.5) -- (401,170.5) -- (401,192.5) -- (377,192.5) -- cycle ;
\draw   (440,170.5) -- (464,170.5) -- (464,192.5) -- (440,192.5) -- cycle ;

\draw (330,53) node [anchor=center][inner sep=0.75pt]  [font=\small]  {$\{1,2,3,4,5\} \ $};
\draw (277.5,94.5) node[font=\small] {$\{1,2\}$};
\draw (379,94.5) node[font=\small] {$\{3,4,5\}$};
\draw (239,138.5) node[font=\small] {$\{1\}$};
\draw (310,139.5) node[font=\small] {$\{2\}$};
\draw (348,139.5) node[font=\small] {$\{3\}$};
\draw (418,139.5) node[font=\small] {$\{4,5\}$};
\draw (389,181.5) node[font=\small] {$\{4\}$};
\draw (452,181.5) node[font=\small] {$\{5\}$};
\end{tikzpicture}
\caption{An example of the HTD binary tree of a fifth-order tensor.}\label{fig: htdtree}
\end{figure}
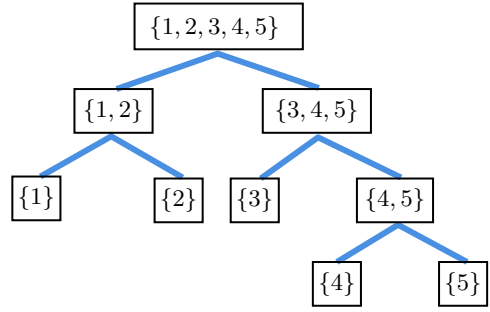

\subsubsection{Canonical Polyadic Decomposition}
Canonical polyadic  decomposition (CPD) is a fundamental tensor decomposition. It represents a tensor as a sum of rank-one tensors, which can be viewed as
higher-order generalization of  matrix eigenvalue decomposition, see Fig.~\ref{fig: CPD}.  CPD preserves the intrinsic multi-way structure of the data and represents interactions across all modes simultaneously.
The CPD of a $k$th-order tensor
$\mathscr A \in \mathbb{R}^{n_1 \times n_2 \times \cdots \times n_k}$ takes the form
\begin{equation}\label{eq:cpd}
\mathscr A=
\sum_{j=1}^{r}
\lambda_j \,
\textbf u^{(1)}_{j}
\circ
\textbf u^{(2)}_{j}
\circ
\cdots
\circ
\textbf u^{(k)}_{j},
\end{equation}
where $r$ is the CP-rank defined as the minimum integer such that
the decomposition \eqref{eq:cpd} holds, $\lambda_j \in \mathbb{R}_+$ are scalar weights, and $\textbf U^{(p)} = [\textbf u^{(p)}_1 \ \textbf u^{(p)}_2\ \cdots \ \textbf u^{(p)}_r] \in \mathbb R^{n_p\times r}$ are factor matrices whose
columns have unit norm. Each rank-one term corresponds to a separable multilinear component that captures a coherent interaction pattern across all tensor modes, with the weight $\lambda_j$ quantifying its relative contribution.

A key property of CPD is that it admits a compact matricized representation. Specifically, the mode-$p$ unfolding $\textbf A_{(p)}$ can be written as
\begin{equation*}
\scalebox{0.87}{$
\textbf A_{(p)}
=\textbf U^{(p)}
\boldsymbol{\Lambda}
\left(
\textbf U^{(k)}
\odot
\cdots
\odot
\textbf U^{(p+1)}
\odot
\textbf U^{(p-1)}
\odot
\cdots
\odot
\textbf U^{(1)}
\right)^{\top},$}
\end{equation*}
where $\boldsymbol{\Lambda}=\mathrm{diag}(\lambda_1,\lambda_2,\dots,\lambda_r)$
is a diagonal matrix containing the weights $\lambda_j$. This formulation underlies many numerical algorithms for computing CPD, including alternating least squares, by reducing the tensor decomposition problem to a sequence of structured linear least-squares subproblems. CPD is attractive due to its conceptual simplicity, interpretability, and high compactness, and it is essentially unique up to permutation and scaling under mild conditions. Although the best low-rank CP approximation problem is ill-posed in general, truncating the CP-rank often yields accurate and practically useful approximations.

\begin{figure}[t]
\centering

\tikzset{every picture/.style={line width=0.75pt}} 

\begin{tikzpicture}[x=0.75pt,y=0.75pt,yscale=-1,xscale=1]

\draw  [color={rgb, 255:red, 74; green, 144; blue, 226 }  ,draw opacity=0.5 ][fill={rgb, 255:red, 74; green, 144; blue, 226 }  ,fill opacity=0.55 ] (88.6,91.51) -- (133.06,91.51) -- (133.06,130.55) -- (88.6,130.55) -- cycle ;
\draw  [color={rgb, 255:red, 74; green, 144; blue, 226 }  ,draw opacity=0.2 ][fill={rgb, 255:red, 74; green, 144; blue, 226 }  ,fill opacity=0.25 ] (107.54,74.03) -- (151.74,74.03) -- (132.8,91.51) -- (88.6,91.51) -- cycle ;
\draw  [color={rgb, 255:red, 74; green, 144; blue, 226 }  ,draw opacity=0.3 ][fill={rgb, 255:red, 74; green, 144; blue, 226 }  ,fill opacity=0.4 ] (152,113.68) -- (152,74.62) -- (133.32,91.36) -- (133.32,130.42) -- cycle ;
\draw  [color={rgb, 255:red, 65; green, 117; blue, 5 }  ,draw opacity=0.3 ][fill={rgb, 255:red, 65; green, 117; blue, 5 }  ,fill opacity=0.3 ] (181.51,91.51) -- (185.25,91.51) -- (185.25,130.42) -- (181.51,130.42) -- cycle ;
\draw  [color={rgb, 255:red, 248; green, 231; blue, 28 }  ,draw opacity=0.5 ][fill={rgb, 255:red, 245; green, 166; blue, 35 }  ,fill opacity=0.4 ] (213.42,61.51) -- (215.65,64.12) -- (185.2,85.4) -- (182.97,82.79) -- cycle ;
\draw  [color={rgb, 255:red, 74; green, 74; blue, 74 }  ,draw opacity=0.2 ][fill={rgb, 255:red, 74; green, 74; blue, 74 }  ,fill opacity=0.3 ] (190.92,89.77) -- (225,89.77) -- (225,93) -- (190.92,93) -- cycle ;
\draw  [color={rgb, 255:red, 65; green, 117; blue, 5 }  ,draw opacity=0.3 ][fill={rgb, 255:red, 65; green, 117; blue, 5 }  ,fill opacity=0.3 ] (254.92,92.1) -- (258.65,92.1) -- (258.65,131) -- (254.92,131) -- cycle ;
\draw  [color={rgb, 255:red, 248; green, 231; blue, 28 }  ,draw opacity=0.5 ][fill={rgb, 255:red, 245; green, 166; blue, 35 }  ,fill opacity=0.4 ] (286.83,62.09) -- (289.06,64.71) -- (258.61,85.98) -- (256.38,83.37) -- cycle ;
\draw  [color={rgb, 255:red, 74; green, 74; blue, 74 }  ,draw opacity=0.2 ][fill={rgb, 255:red, 74; green, 74; blue, 74 }  ,fill opacity=0.3 ] (264.32,90.35) -- (299,90.35) -- (299,94) -- (264.32,94) -- cycle ;
\draw  [color={rgb, 255:red, 65; green, 117; blue, 5 }  ,draw opacity=0.3 ][fill={rgb, 255:red, 65; green, 117; blue, 5 }  ,fill opacity=0.3 ] (358.16,91.51) -- (361.9,91.51) -- (361.9,130.42) -- (358.16,130.42) -- cycle ;
\draw  [color={rgb, 255:red, 248; green, 231; blue, 28 }  ,draw opacity=0.5 ][fill={rgb, 255:red, 245; green, 166; blue, 35 }  ,fill opacity=0.4 ] (390.08,61.51) -- (392.31,64.12) -- (361.86,85.4) -- (359.63,82.79) -- cycle ;
\draw  [color={rgb, 255:red, 74; green, 74; blue, 74 }  ,draw opacity=0.2 ][fill={rgb, 255:red, 74; green, 74; blue, 74 }  ,fill opacity=0.3 ] (367.57,89.77) -- (400,89.77) -- (400,93) -- (367.57,93) -- cycle ;

\draw (158.29,84.38) node [anchor=north west][inner sep=0.75pt]   [align=left] {=};
\draw (103.44,105.32) node [anchor=north west][inner sep=0.75pt]    {$\mathscr A$};
\draw (232.56,82.64) node [anchor=north west][inner sep=0.75pt]   [align=left] {+};
\draw (304.75,80.89) node [anchor=north west][inner sep=0.75pt]   [align=left] {+ };
\draw (316.5,85.02) node [anchor=north west][inner sep=0.75pt]   [align=left] {{\LARGE ... }};
\draw (340.68,81.47) node [anchor=north west][inner sep=0.75pt]   [align=left] {+ };

\end{tikzpicture}
\caption{An example of the CPD of a third-order tensor.}\label{fig: CPD}
\end{figure}
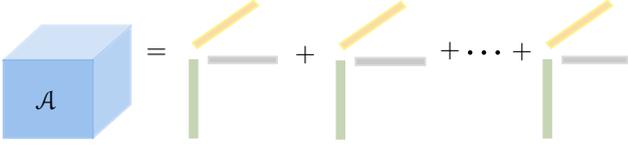

  \section{Data-Driven Tensor Decomposition Identification of HPDSs}\label{sec: identification}
Any homogeneous polynomial dynamical system (HPDS) of degree $k-1$ can be expressed compactly in terms of tensor-vector multiplications as
\begin{align}\label{eq: ausys}
\dot{\textbf x}(t)=\mathscr A\textbf x(t)^{k-1},
\end{align}
where $\mathscr A\in\mathbb{R}^{n\times n\times\stackrel{k}{\cdots}\times n}$ is a $k$th-order, $n$-dimensional almost symmetric tensor, and $\textbf x(t)\in\mathbb{R}^n$ denotes the system state. This tensor-based representation provides a concise and structured way to capture higher-order interactions among the state variables. To illustrate, consider the following quadratic HPDS
\[
\dot{x}_1 = x_1^2 - 3x_1x_2 + 2x_2^2,\quad
\dot{x}_2 = 2x_1^2 + 6x_1x_2 - x_2^2.
\]
 This system can be written compactly as $\dot{\textbf x}(t)=\mathscr A \textbf x(t)^{2}$,
where $\mathscr A\in\mathbb R^{2\times 2\times 2}$ is a third-order almost symmetric dynamic tensor with
frontal slices
\[
\mathscr{A}_{::1} =
\begin{bmatrix}
1 & -\tfrac{3}{2} \\
-\tfrac{3}{2}& 2
\end{bmatrix},\quad
\mathscr{A}_{::2} =
\begin{bmatrix}
2 & 3 \\
3 & -1
\end{bmatrix}.
\]
We assume that state measurements are collected over the time interval $[t_0, t_0+(T-1)\tau]$ with a uniform sampling period $\tau$. The sampled state trajectory and its time derivative are organized into the data matrices
\begin{align*}
\textbf X_0&=\begin{bmatrix}
\textbf x(t_0) & \textbf x(t_0+\tau) & \cdots & \textbf x(t_0+(T-1)\tau)\end{bmatrix}\in\mathbb{R}^{n\times T},\\
\textbf X_1&=\begin{bmatrix}
\dot{\textbf x}(t_0) & \dot{\textbf x}(t_0+\tau) & \cdots & \dot{\textbf x}(t_0+(T-1)\tau)\end{bmatrix}\in\mathbb{R}^{n\times T}.
\end{align*}
The time-series data $\textbf X_0$ and $\textbf X_1$  serve as the basis for identifying system parameters directly from observations. For notational convenience, 
we denote $\textbf x(j)$ as the $j$th sampled state,
i.e., $\textbf x(j)=\textbf x(t_0+(j-1)\tau)$.

While the tensor-based representation \eqref{eq: ausys} provides a compact description of HPDSs, the direct identification of the full dynamic tensor $\mathscr A$ from data is generally
challenging due to the extremely high dimensionality and combinatorial complexity of $\mathscr A$.  This challenge becomes particularly severe for high-degree and high-dimensional systems, where naive least-squares formulations quickly become ill-posed or computationally prohibitive. To address this issue, we exploit the observation that many real-world higher-order systems exhibit intrinsic low-rank structure when represented in suitable tensor formats. By imposing structured tensor decompositions on $\mathscr A$, the original identification problem can be reformulated in terms of a substantially smaller set of factor matrices or factor tensors. This leads to a parsimonious representation of the underlying dynamics and enables the development of efficient alternating optimization algorithms. In the following, we develop data-driven tensor decomposition identification methods for HPDSs. All theoretical results are detailed based on the TTD-based representation, while the HTD- and CPD-based representations follow analogously.

\subsection{Identification of TTD-based HPDSs}

We begin by studying the identification of TTD-based HPDSs. In this representation, the dynamic tensor $\mathscr A$ is factorized into a sequence of third-order factor tensors whose dimensions are governed by the TT-ranks $\{r_p\}_{p=0}^{k}$. Assume $r=\text{max}\{r_p\ | \ p=0,1,\dots,k\}$. The total number of parameters in the TTD-based representation is $\mathcal O(knr^2)$, which is linearly scaling in the order $k$ and quadratic scaling in the maximum TT-rank $r$. This parameterization yields a substantial reduction in complexity compared with the exponential scaling of the full tensor representation, making the identification of large-scale HPDSs computationally tractable. Our objective is to identify factor tensors associated with the TTD of $\mathscr{A}$ that best describe the system dynamics from the observed time-series data $\textbf X_0$ and $\textbf X_1$. 
     
To connect the TTD-based representation with the observed data, we first express the homogeneous polynomial vector field in terms of the factor tensors. For each sampled state $\textbf x(j)$, we define 
\begin{equation}\label{eq: sysTTD}
    \textbf f(j)=\left[\prod_{p=1}^{k-1}\left(\mathscr T^{(p)}\times_2 \textbf x(j)\right)\mathscr T^{(k)}\right]^\top,
\end{equation}
and form the matrix $\textbf F_0=[
\textbf f(0) \ \textbf f(1) \ \cdots \ \textbf f(T-1)]\in\mathbb{R}^{n\times T}.$
The identification of factor tensors is formulated as the following least-squares optimization problem
    \begin{align}\label{eq:ttd-opt}
\min_{\{\mathscr{T}^{(p)}\}_{p=1}^k}  
\| \textbf{X}_1 -\textbf F_0 \|_F^2,
\end{align}
which is nonlinear and nonconvex due to the recursive contractions across the factor tensors. Nevertheless, by fixing all factor tensors except $\mathscr T^{(p)}$, the problem becomes linear in $\mathscr T^{(p)}$, allowing it to be solved as a standard linear least-squares subproblem. This observation naturally motivates an alternating least-squares (ALS) scheme, in which each factor tensor is updated sequentially while keeping the others fixed, resulting in a sequence of tractable linear problems. 

The detailed procedure is summarized in Algorithm~\ref{alg:ALS-TTD}. 
At each iteration, the factor tensors
$\{\mathscr T^{(p)}\}_{p=1}^k$ are updated sequentially by solving a regression problem constructed from the available data. For the $p$th factor tensor $\mathscr{T}^{(p)}$ and time index $j$, contracting the remaining fixed factor tensors with the sampled state $\textbf x(j)$ yields the left and right matrices
\begin{align}
\textbf{L}_p(j) &= \prod_{i=1}^{p-1} \left( \mathscr{T}^{(i)} \times_2 \textbf{x}(j) \right),\label{eq: ttdleft} \\
\textbf{R}_p(j) &= \prod_{i=p+1}^{k-1} \left( \mathscr{T}^{(i)} \times_2 \textbf{x}(j) \right) \mathscr{T}^{(k)}.\label{eq: ttdright}
\end{align}
Both $\textbf{L}_p(j)$ and $\textbf{R}_p(j)$ are low-dimensional matrices whose sizes depend only on the TT-ranks and the state dimension $n$, and can be efficiently computed via sequential tensor contractions. 

By leveraging these contractions, the update of $\mathscr T^{(p)}$ reduces to the linear least-squares problem
\begin{align}\label{eq:ttd-sub}
\min_{\mathrm{vec}(\mathscr T^{(p)})}
\| \textbf H_p\,\mathrm{vec}(\mathscr T^{(p)}) - \mathrm{vec}(\textbf X_1) \|_2^2,
\end{align}
where $\textbf H_p$
is constructed by stacking the matrices $\textbf H_p = [\textbf H_p(0)^\top \  \textbf H_p(1)^\top \  \cdots \  \textbf H_p(T-1)^\top]^\top$
with each block
\begin{align}\label{eq: ttdH}
  \textbf H_p(j) = \textbf R_p(j)^\top \otimes \textbf x(j)^\top \otimes \textbf L_p(j).
  \end{align}
The matrix $\textbf H_p$ has dimension $(nT) \times (nr_{p-1}r_p)$. Therefore, the subproblem avoids
the exponential complexity associated with identifying the full high-dimensional dynamic tensor in the ambient tensor space. The solution is reshaped into a third-order tensor of dimension $r_{p-1}\times n\times r_p$, followed by an orthogonalization step to maintain numerical stability and control scaling between adjacent factor tensors. The TT-ranks may either be specified a priori or adjusted adaptively via SVD-based truncation during the orthogonalization stage. 
The algorithm terminates when the relative decrease in prediction error falls below a prescribed tolerance, or when the maximum number of iterations is attained. This structured ALS scheme therefore provides a computationally tractable and scalable approach for identifying TTD-based HPDSs from time-series data.

\begin{algorithm}[t]
\caption{ALS Identification for TTD-based HPDS}
\label{alg:ALS-TTD}
\begin{algorithmic}[1]
\State \textbf{Input:} Data $\textbf X_0, \textbf X_1$, tolerance $\epsilon > 0$, truncation threshold $\delta > 0$, maximum iterations $N_{\max}$
\State \textbf{Output:} TT factor tensors $\{\mathscr T^{(p)}\}_{p=1}^k$ 
\State Randomly initialize $\{\mathscr T^{(p)}\}_{p=1}^k$ ; set $e_{\text{prev}} = +\infty$
\For{$\text{iter} = 1$ to $N_{\max}$}
    \For{$p = 1$ to $k$}
        \For{$j = 1$ to $T$}
            \State Compute left contraction \eqref{eq: ttdleft} 
            \State Compute right contraction \eqref{eq: ttdright}
            \State Construct  $\textbf H_p(j)$ based on \eqref{eq: ttdH}
        \EndFor
        \State Stack $\textbf H_p = [\textbf H_p(0)^\top\ \cdots\ \textbf H_p(T-1)^\top]^\top$
        \State Solve the linear least-squares problem \eqref{eq:ttd-sub}
        \State Reshape solution into $\mathscr T^{(p)} \in \mathbb{R}^{r_{p-1}\times n\times r_p}$
        \State Orthonormalize $\mathscr T^{(p)}$ and truncate singular values below $\delta$ if necessary
    \EndFor
    \State Compute $\textbf F_0 = [\textbf f(0)\ \textbf f(1) \ \cdots\ \textbf f(T-1)]$ using \eqref{eq: sysTTD}
    \State Compute error $e = \|\textbf X_1 - \textbf F_0\|_F^2$
    \If{$|e_{\text{prev}} - e|/e_{\text{prev}} < \epsilon$}
        \State \textbf{break}
    \EndIf
    \State $e_{\text{prev}} \gets e$
\EndFor
\State \textbf{return}$\{\mathscr T^{(p)}\}_{p=1}^k$ 
\end{algorithmic}
\end{algorithm}

\begin{remark}
Let $r$ denote the maximum TT-rank of the $k$th-order, $n$-dimensional dynamic tensor $\mathscr{A}$ associated with a TTD-based HPDS. The computational complexity of a single ALS iteration in Algorithm~\ref{alg:ALS-TTD} can then be estimated as $\mathcal O(k^2nr^2T+kn^3r^4T)$.
\end{remark}

Next, we analyze the convergence properties of Algorithm~\ref{alg:ALS-TTD}, including the monotonic decrease of the objective function, as well as the local convergence under a mild regularity condition. Let $\{\{\mathscr T^{(p)}_\ell\}_{p=1}^k\}_{\ell\ge 0}$ denote the sequence of  factor tensors generated by Algorithm~\ref{alg:ALS-TTD}, where $\ell$ denotes the iteration index. 

\begin{proposition}\label{prop:ttd1}
The sequence of objective values of \eqref{eq:ttd-opt} is monotonically non-increasing and convergent. Moreover, every accumulation point of the iterate sequence $\{\{\mathscr T^{(p)}_\ell\}_{p=1}^k\}_{\ell\ge 0}$ is a first-order stationary point of \eqref{eq:ttd-opt}.
\end{proposition}
\textbf{Proof.}
At each update of a single factor tensor while keeping the other factor tensors fixed, Algorithm~\ref{alg:ALS-TTD} exactly minimizes a convex quadratic least-squares subproblem. Hence the objective value of \eqref{eq:ttd-opt} cannot increase after each block update, implying that the sequence of objective values is monotonically non-increasing over successive sweeps. Since the objective function is nonnegative, the sequence of objective values is bounded below and therefore convergent. Moreover, the objective function is a polynomial function of the factor tensors and is therefore semi-algebraic. Semi-algebraic functions satisfy the Kurdyka--{\L}ojasiewicz (KL) property. Standard convergence results for block coordinate descent methods applied to KL functions imply that every accumulation point of the iterate sequence is a first-order stationary point of \eqref{eq:ttd-opt}. \hfill $\blacksquare$

\begin{proposition}\label{prop:ttd2}
Let $\{\mathscr T^{(p)}_\star\}_{p=1}^{k}$ be an isolated stationary point of \eqref{eq:ttd-opt}.
Assume that there exists a neighborhood $\mathcal N$ of $\{\mathscr T^{(p)}_\star\}_{p=1}^{k}$ such that for all iterates in $\mathcal N$ and $p=1,2,\dots,k$,
the corresponding regression matrix $\textbf H_p$
satisfies
$
\textbf H_p^\top \textbf H_p \succeq \alpha_p \textbf I$
for some constant $\alpha_p>0$. Then Algorithm~\ref{alg:ALS-TTD} converges locally linearly to
$\{\mathscr T^{(p)}_\star\}_{p=1}^{k}$, i.e., there exists   a neighborhood $\mathcal N'\subseteq \mathcal N$ of
$\{\mathscr T^{(p)}_\star\}_{p=1}^k$ such that for all iterates in $\mathcal{N}'$, it holds that
\begin{equation}\label{eq:localconv}
\sum_{p=1}^k
\|\mathscr T^{(p)}_{\ell+1}-\mathscr T^{(p)}_\star\|_F^2
\le\rho\sum_{p=1}^k
\|\mathscr T^{(p)}_{\ell}-\mathscr T^{(p)}_\star\|_F^2
\end{equation}
for some $\rho\in(0,1)$.
\end{proposition}
\textbf{Proof.} By fixing an index $p$ and all  factor tensors
except $\mathscr T^{(p)}$, the mapping from
$\mathscr T^{(p)}$ to $\textbf F_0$ is linear and the objective
reduces to the least-squares problem \eqref{eq:ttd-sub}. By assumption, there exists a neighborhood $\mathcal N$ of
$\{\mathscr T^{(p)}_\star\}_{p=1}^k$ such that, for all iterates in $\mathcal N$
and all $p=1,2,\ldots,k$, $\textbf H_p^\top\textbf H_p\succeq \alpha_p \textbf I
\quad\text{for some }\alpha_p>0.$
Therefore, each subproblem has a unique minimizer
$$\mathrm{vec}(\mathscr T^{(p)}_{+})
=
(\textbf H_p^\top\textbf H_p)^{-1}\textbf H_p^\top \mathrm{vec}(\textbf X_1),$$
with $\|(\textbf H_p^\top\textbf H_p)^{-1}\|\le 1/\alpha_p.$ Thus, updating the $p$th factor tensor defines a single-valued mapping from the other factor tensors to $\mathscr T^{(p)}_{+}$.

Moreover, $\textbf H_p$ depends on the other factor tensors through a finite sequence of TTD contractions and is therefore a polynomial mapping.
Hence, $\textbf H_p$ is locally Lipschitz on $\mathcal N$. Together with the uniform bound on $(\textbf H_p^\top\textbf H_p)^{-1}$, this
implies that the update for each factor tensor is locally Lipschitz. Consequently, the full ALS sweep defines a locally Lipschitz mapping $\Phi$ that admits $\{\mathscr T^{(p)}_\star\}_{p=1}^k$ as a fixed point. Under the strong regularity condition and after fixing the gauge freedom via orthonormalization, the stationary point is isolated and satisfies a local second-order regularity condition. Under these conditions, results on nonlinear Gauss–Seidel methods implies that the Jacobian $\nabla \Phi$ at $\{\mathscr T^{(p)}_\star\}_{p=1}^k$ corresponds to a locally stable fixed point, i.e., its spectral radius is strictly less than one. By continuity of $\nabla \Phi$, there exists a sufficiently small neighborhood $\mathcal N'\subseteq\mathcal N$ such that $\Phi$ is contractive on $\mathcal N'$. Hence, there exists $\rho\in(0,1)$ such that \eqref{eq:localconv} holds,
which proves the claimed local linear convergence. \hfill$\blacksquare$




Propositions~\ref{prop:ttd1} and \ref{prop:ttd2} characterize the convergence behavior of Algorithm~\ref{alg:ALS-TTD}. 
They show that the algorithm is both globally well-behaved and locally efficient, providing theoretical support for its use in identifying TTD-based HPDSs. We then analyze the robustness of the algorithm in the presence of measurement noise. Consider the TTD-based HPDS corrupted by additive noise 
\begin{equation}
    \dot{\textbf x}(t)=\left[\prod_{p=1}^{k-1}(\mathscr T^{(p)}\times_2\textbf x(t))\mathscr T^{(k)}\right]^\top+\textbf w(t),
\end{equation}
where $\textbf w(t)$ is an independent, identically distributed (i.i.d.) noise process. After sampling, instead of the exact data pair $(\textbf X_0,\textbf X_1)$, we observe noisy data
$\tilde{\textbf X}_1 = \textbf X_1 + \textbf W,$ where $\textbf W\in\mathbb R^{n\times T}$ is a measurement noise matrix. At each iteration, the update of the $p$th factor tensor is to solve the linear least-squares problem
\begin{equation}
    \min_{\mathrm{vec}(\mathscr T^{(p)})}\| \textbf H_p\,\mathrm{vec}(\mathscr T^{(p)}) - \mathrm{vec}(\tilde{\textbf X}_1)\|_2^2.
\end{equation}
We denote the obtained factor tensor at $\ell$th iteration by $\{\hat{\mathscr T}_\ell^{(p)}\}_{p=1}^k$.
The estimation error of a single block update can be bounded as follows.
\begin{proposition}
Assume that the regression matrix $\textbf{H}_p$ satisfies $\textbf H_p^\top \textbf H_p \succeq \alpha_p \textbf{I}$
for some constant $\alpha_p>0$. 
Then the least-squares solution satisfies 
\begin{align}\label{eq: noisedeter}
\|\hat{\mathscr T}_\ell^{(p)}- \mathscr T_\ell^{(p)}\|_F\le
\frac{1}{\sqrt{\alpha_p}}\|\textbf W\|_F.
\end{align}
Moreover, if the entries of $\textbf W$ are i.i.d. sub-Gaussian with variance proxy $\sigma^2$, 
then for any $\delta\in(0,1)$, with probability at least $1-\delta$, it holds that
\begin{equation}\label{eq: noisesta}
\|
\hat{\mathscr T}_{\ell}^{(p)}-{\mathscr T}_{\ell}^{(p)}
\|_F
\le
\frac{c\sigma}{\sqrt{\alpha_p}}\,
\sqrt{nT+\log\frac{1}{\delta}},
\end{equation}
where $c>0$ is an absolute constant. 
\end{proposition}
\textbf{Proof.}
The least-squares solutions of the noisy and noise-free cases satisfy the corresponding normal equations. Subtracting them yields
\[
\textbf H_p^\top \textbf H_p \,
\mathrm{vec}(
\hat{\mathscr T}_\ell^{(p)} - \mathscr T_\ell^{(p)}
)
=
\textbf H_p^\top \mathrm{vec}(\textbf W).
\]
Hence, $\mathrm{vec}(
\hat{\mathscr T}_\ell^{(p)} - \mathscr T_\ell^{(p)})
=
(\textbf H_p^\top \textbf H_p)^{-1}
\textbf H_p^\top
\mathrm{vec}(\textbf W).$
Taking norms gives
\[
\|
\mathrm{vec}(
\hat{\mathscr T}_\ell^{(p)} - \mathscr T_\ell^{(p)}
)
\|_2
\le
\|
(\textbf H_p^\top \textbf H_p)^{-1}
\textbf H_p^\top
\|_2
\|\textbf W\|_F.
\]
Since $\textbf H_p^\top \textbf H_p \succeq \alpha_p \textbf{I}$, it follows that
$\|
(\textbf H_p^\top \textbf H_p)^{-1}
\textbf H_p^\top\|_2
\le 1/\sqrt{\alpha_p}$,
which yields \eqref{eq: noisedeter}.

For the high-probability bound, since the entries of $\textbf W$ are i.i.d. sub-Gaussian with variance proxy $\sigma^2$, standard concentration results imply that for any $\delta\in(0,1)$,
\[
\|\textbf W\|_F
\le
c\sigma\sqrt{nT+\log\frac{1}{\delta}}\]
with probability at least $1-\delta$. Substituting into the previous bound yields \eqref{eq: noisesta}.
\hfill$\blacksquare$

This result shows that each ALS update is stable under measurement noise. We next combine this property with the local contraction of the noiseless iteration to characterize the overall convergence behavior of the algorithm in the noisy setting.
\begin{proposition}
Assume the conditions of Proposition~\ref{prop:ttd2} hold and the entries of $\textbf W$ are i.i.d. sub-Gaussian with variance proxy $\sigma^2$.
Then for any $\delta\in(0,1)$, with probability at least $1-\delta$, there exist
$\bar\rho\in(0,1)$ and a constant $C>0$ such that
\begin{equation*}\label{eq:noise-recursion}
\scalebox{0.87}{$\displaystyle\sum_{p=1}^k\|\hat{\mathscr T}^{(p)}_{\ell+1}-\mathscr T^{(p)}_\star\|_F^2  \le
\bar\rho
\sum_{p=1}^k\|\hat{\mathscr T}^{(p)}_{\ell}-\mathscr T^{(p)}_\star\|_F^2+
C\sigma^2\Big(nT+\log\tfrac{1}{\delta}\Big),$}
\end{equation*}
for iterates sufficiently close to $\{\mathscr T^{(p)}_\star\}_{p=1}^k$.
Consequently,
\begin{equation*}\label{eq:noise-floor}
\limsup_{\ell\to\infty}
\sum_{p=1}^k\|\hat{\mathscr T}^{(p)}_{\ell}-\mathscr T^{(p)}_\star\|_F^2
\;\le\;
\frac{C}{1-\bar\rho}\,
\sigma^2\Big(nT+\log\tfrac{1}{\delta}\Big),
\end{equation*}
showing that the estimation error remains within a neighborhood whose size scales with the noise level.
\end{proposition}

\textbf{Proof.} 
Denote the noiseless error by
$S_\ell=\sum_{p=1}^k\|\mathscr T^{(p)}_{\ell}-\mathscr T^{(p)}_\star\|_F^2.$
By Proposition~\ref{prop:ttd2}, there exist $\rho\in(0,1)$ and a neighborhood $\mathcal N'$
such that for all iterates in $\mathcal N'$, $S_{\ell+1}\le \rho S_\ell.$
Define the noise-induced error
$
D_\ell=\sum_{p=1}^k\|\hat{\mathscr T}^{(p)}_{\ell}-\mathscr T^{(p)}_{\ell}\|_F^2.$
Applying \eqref{eq: noisesta} to each block update with failure probability $\delta/k$ and using a union bound over $p=1,2,\dots,k$, we obtain that with probability at least $1-\delta$, it holds that
\begin{equation}\label{eq:D-bound}
D_\ell \le C_0\,\sigma^2\Big(nT+\log\frac{1}{\delta}\Big),
\end{equation}
for all $\ell$ such that the iterates remain in $\mathcal N'$.

Using the inequality $$\|\textbf A+\textbf B\|_F^2 \le (1+\eta)\|\textbf A\|_F^2+(1+\eta^{-1})\|\textbf B\|_F^2$$ for any $\eta>0$, we decompose
$\hat{\mathscr T}^{(p)}_{\ell}-\mathscr T^{(p)}_\star
=(\hat{\mathscr T}^{(p)}_{\ell}-\mathscr T^{(p)}_{\ell})
+
\big(\mathscr T^{(p)}_{\ell}-\mathscr T^{(p)}_\star\big)$
and obtain
\begin{equation}\label{eq:E-upper-2}
E_\ell=\sum_{p=1}^k\|\hat{\mathscr T}^{(p)}_{\ell}-\mathscr T^{(p)}_\star\|_F^2
\le
(1+\eta)D_\ell+(1+\eta^{-1})S_\ell.
\end{equation}
Similarly, we have
\begin{equation}\label{eq:S-upper-2}
S_\ell \le (1+\eta)D_\ell+(1+\eta^{-1})E_\ell.
\end{equation}
Since $S_{\ell+1}\le \rho S_\ell$, combining \eqref{eq:E-upper-2}, and \eqref{eq:S-upper-2} yields
\[
E_{\ell+1}
\le
(1+\eta)D_{\ell+1}
+
(1+\eta^{-1})\rho\big((1+\eta)D_\ell+(1+\eta^{-1})E_\ell\big).
\]
Using \eqref{eq:D-bound} to bound $D_\ell$ and $D_{\ell+1}$, we obtain
\[
E_{\ell+1}\le \bar\rho E_\ell
+C\sigma^2\Big(nT+\log\tfrac{1}{\delta}\Big),
\]
where $\bar\rho=\rho(1+\eta^{-1})^2$. By choosing $\eta$ sufficiently large, we ensure that $\bar\rho\in(0,1)$.
Taking $\limsup_{\ell\to\infty}$ yields the desired bound.
\hfill$\blacksquare$

The above result shows that, in the presence of noise, the ALS algorithm converges linearly up to a neighborhood of the true solution, and the asymptotic error is bounded by a noise-dependent term.
If the state measurements $\textbf X_0$ are also corrupted by noise, 
then the regression matrices $\textbf H_p$ become noisy, leading to an errors-in-variables problem. 
In this case the least-squares estimator is generally biased. 
A consistent alternative is to adopt an integral formulation 
of the dynamics, which avoids numerical differentiation and 
reduces noise amplification. 
A detailed analysis of this scenario is left for future work.
Beyond convergence and noise robustness, it is essential to characterize the data informativity conditions that guarantee identifiability 
under the TTD parameterization of HPDSs. 

\begin{proposition}\label{prop:ttdinfo}
A sufficient condition for identification of TTD-based HPDSs
from data is
\begin{align}\label{eq:aurankcondition}
\scalebox{0.9}{$\displaystyle\mathrm{rank}(\hat{\textbf X_0})
=
\sum_{j=1}^{\min\{n,k-1\}}
\frac{n!}{j!(n-j)!}
\frac{(k-2)!}{(j-1)!(k-j-1)!},$}
\end{align}
where
$\hat{\textbf X_0} = \textbf X_0 \odot \textbf X_0 \odot \stackrel{k-1}\cdots \odot \textbf X_0$. A necessary condition  is that the  stacked regression matrices 
$\textbf H_p$ have full column rank for $p=1,2,\dots,k$.
\end{proposition}

\textbf{Proof.} We analyze the data informativity from two complementary perspectives. From \cite{mao2025tensor}, the data uniquely determine the full polynomial tensor $\mathscr A$ if and only if \eqref{eq:aurankcondition} holds.
Since the TTD representation defines a parameterization of $\mathscr A$, uniqueness of $\mathscr A$ implies identifiability of the corresponding factor tensors (up to gauge transformations). Therefore, \eqref{eq:aurankcondition} provides a sufficient condition for identifying TTD-based HPDSs. 
Under the TTD parameterization, identification reduces to the sequence of linear least-squares problems \eqref{eq:ttd-sub}
for $p=1,2,\dots,k$. For each subproblem to admit a unique solution, the stacked regression matrices $\textbf H_p$ must have full column rank. Otherwise, multiple factor tensors produce identical trajectories, and identification is not unique. \hfill$\blacksquare$

Condition \eqref{eq:aurankcondition} guarantees identifiability of the factor tensors through uniqueness of the underlying dynamic tensor $\mathscr A$. 
However, the TTD parameterization constrains $\mathscr A$ to a low-dimensional nonlinear manifold with only $\mathcal O(kn r^2)$ degrees of freedom. Consequently, \eqref{eq:aurankcondition}, which characterizes identifiability in the ambient polynomial space, is generally sufficient but not minimal for identification under the TTD model. In particular, identifiability of low-rank TTD-based HPDSs can be achieved under weaker excitation conditions that ensure injectivity of the data mapping restricted to the TTD manifold. On the other hand, the necessary condition that each regression matrix $\mathbf H_p$ has full column rank guarantees uniqueness of the individual least-squares subproblems, but does not ensure global identifiability of the full TTD representation. The
 reason is that the TTD parameterization is nonlinearly coupled across different factor tensors, and the regression matrices themselves depend on the remaining factor tensors.


\subsection{Identification of HTD-based HPDSs}

We next investigate the problem under the HTD-based representation. In contrast to the sequential structure of  TTD , the HTD-based representation organizes the tensor modes according to a hierarchical binary tree. Let $r=\text{max}\{r_\mathcal Q \text{ }|\text{ } \mathcal Q\in\mathcal T\}$ denote the maximum hierarchical rank over all nodes in the dimension tree $\mathcal T$. The total number of parameters in the HTD-based representation can be estimated as $\mathcal O(knr+kr^3)$, which is substantially smaller than that of the full tensor representation when $r$ is moderate. The objective is to identify all factor matrices $\{\textbf V_p\}_{p=1}^k$ for the leaf nodes and transfer matrices $\{\textbf C_\mathcal P\}_{\mathcal P\in\mathcal T}$ for the internal nodes.
   
 Under HTD, the system mapping associated with the sampled state $\textbf x(j)$ can be written as 
\begin{align}\label{eq: sysHTD}
\textbf f(j)=&\left(\textbf V_k\otimes\textbf x(j)^\top\textbf V_{k-1}\otimes\cdots\otimes\textbf x(j)^\top\textbf V_{1}\right)\nonumber\\
&(\otimes_{\mathcal Q\in\mathcal G_{d-1}}\textbf C_\mathcal Q)\cdots(\otimes_{\mathcal Q\in\mathcal G_{0}}\textbf C_\mathcal Q),
\end{align}
which forms the predicted output matrix $\textbf F_0$. The  identification problem can therefore be written as
 \begin{align}\label{eq:htd-opt}
\min_{\{\textbf V_{p}\}_{p=1}^k,\{\textbf C_{\mathcal P}\}_{\mathcal P\in\mathcal T}} \| \textbf{X}_1 -\textbf F_0 \|_F^2.
\end{align}
As in the TTD formulation, the problem is nonconvex due to the recursive tensor contractions along the dimension tree but retains a block multilinear structure. The detailed ALS procedure is summarized in Algorithm~\ref{alg:ALS-HTD}, where the functions $\Call{LeafLS}$ and $\Call{InternalLS}$ can be found in the appendix. 

 A key distinction from  TTD is that the HTD-based representation
admits level-wise updates along the dimension tree. At each iteration, all leaf-node factor matrices are first updated, followed by the updates of the internal-node transfer matrices from the bottom of the tree toward the root. Moreover, nodes located at the same level are independent, and can be updated simultaneously. This yields a parallel sweep implementation of the ALS algorithm. When updating a leaf-node factor matrix $\textbf V_p$,
the HTD contraction can be decomposed into a leaf-side
contraction $\textbf L_p(j)$ and a tree-side contraction
$\textbf R_p(j)$. The leaf-side contraction collects the
contributions from all leaf-node factor matrices except
$\textbf V_p$ and is given by
\begin{align}\label{eq:htdleft}
\textbf L_p(j)
=
\bigotimes_{q=k}^{1}
\begin{cases}
\textbf Z_q(j), & q\neq p,\\
\textbf x(j)^\top, & q=p,
\end{cases}
\end{align}
where $\textbf Z_q(j)=\textbf x(j)^\top\textbf V_q$ for
$q=1,2,\dots,k-1$ and $\textbf Z_k(j)=\textbf V_k$.
The tree-side contraction $\textbf R_p(j)$ efficiently aggregates the
contributions of the internal-node transfer matrices along the
dimension tree and can be obtained recursively by performing
tensor contractions from the lowest internal level toward the root.

 Using these contractions, the update of the $p$th leaf-node factor matrices reduces to the following linear least-squares problem            \begin{align}\label{eq:htd-sub}
            \min_{\mathrm{vec}(\textbf V_p)}\|\textbf H_p\,\mathrm{vec}(\textbf V_p)-\mathrm{vec}(\textbf X_1)\|_2^2.
            \end{align}Here, the regression matrix $\textbf H_p\in\mathbb R^{(nT)\times(nr_p)}$  is obtained by stacking the sample-wise blocks \begin{align} \label{eq:htdH}\textbf H_p(j)
        =
        \big(\textbf R_p(j)^\top\otimes \textbf L_p(j)\big)\textbf S_p,
        \end{align}
        where $\textbf S_p$ is a permutation operator satisfying
        $\mathrm{vec}(\textbf I\otimes\textbf V_p\otimes\textbf I)=\textbf S_p\,\mathrm{vec}(\textbf V_p)$. The least-squares solution is reshaped into the matrix $\textbf V_p \in \mathbb R^{n \times r_p}.$ The update of each internal-node transfer matrix $\textbf C_{\mathcal P}$ is derived analogously by isolating $\mathrm{vec}(\textbf C_{\mathcal P})$ while keeping all remaining parameters fixed during the optimization step. 

\begin{algorithm}[t]
\caption{ALS Identification for HTD-based HPDS }
\label{alg:ALS-HTD}
\begin{algorithmic}[1]
\State \textbf{Input:} Data $\textbf X_0, \textbf X_1$, dimension tree $\mathcal T$ with depth $d$, tolerance $\epsilon>0$, maximum iterations $N_{\max}$
\State \textbf{Output:} Leaf-node factor matrices $\{\textbf V_p\}_{p=1}^k$ and internal-node transfer matrices $\{\textbf C_\mathcal P\}_{\mathcal P\in\mathcal T}$
\State Randomly initialize $\textbf V_p$ for all leaf nodes and $\textbf C_\mathcal P$ for all internal nodes; set $e_{\text{prev}}=+\infty$

\For{$\mathrm{iter}=1$ to $N_{\max}$}

\State $\{\textbf V_p^{\mathrm{old}}\}_{p=1}^k \gets \{\textbf V_p\}_{p=1}^k$
\ForAll{$p=1,2,\dots,k$ \textbf{in parallel}}
    \State \scalebox{0.9}{$\displaystyle\textbf H_p \gets \Call{LeafLS}{p,\textbf X_0,\{\textbf V_p^{\mathrm{old}}\}_{p=1}^k,\{\textbf C_\mathcal P\}_{\mathcal P\in\mathcal T},\mathcal T}$}
    \State Solve the linear least-squares problem \eqref{eq:htd-sub}
    \State Reshape the solution as $\textbf V_p^{\mathrm{new}}\in\mathbb R^{n_p\times r_p}$
\EndFor
\State $\{\textbf V_p\}_{p=1}^k \gets \{\textbf V_p^{\mathrm{new}}\}_{p=1}^k$

\For{$l=d$ to $1$}
    \State $\{\textbf C_\mathcal P^{\mathrm{old}}\}_{\mathcal P\in\mathcal T} \gets \{\textbf C_\mathcal P\}_{\mathcal P\in\mathcal T}$
    \ForAll{internal node $\mathcal P\in\mathcal G_l$ \textbf{in parallel}}
        \State $\textbf H_\mathcal P\leftarrow\textsc{InternalLS}(\mathcal P,l,\textbf X_0,\{\textbf V_p\}_{p=1}^k,$
        \State $\{\textbf C_\mathcal P^{\mathrm{old}}\}_{\mathcal P\in\mathcal T},\mathcal T)$
        \State Solve $\min_{\mathrm{vec}(\textbf C_\mathcal P)}\!\|\textbf H_\mathcal P\mathrm{vec}(\textbf C_\mathcal P\!)-\mathrm{vec}(\textbf X_1\!)\|_2^2$
        \State Store the  solution as $\textbf C_\mathcal P^{\mathrm{new}}\in\mathbb R^{r_{\mathcal P_l}r_{\mathcal P_r}\times r_\mathcal P}$
    \EndFor
    \State $\{\textbf C_\mathcal P\}_{\mathcal P\in\mathcal G_l}\gets \{\textbf C_\mathcal P^{\mathrm{new}}\}_{\mathcal P\in\mathcal G_l}$
\EndFor

\State Compute  $\textbf F_0 = [\textbf f(0)\ \textbf f(1)\ \cdots\ \textbf f(T-1)]$  using \eqref{eq: sysHTD}
\State $e=\|\textbf X_1-\textbf F_0\|_F^2$
\If{$|e-e_{\text{prev}}|<\epsilon$} \textbf{break} \EndIf
\State $e_{\text{prev}}=e$

\EndFor
\State \textbf{return} $\{\textbf V_p\}_{p=1}^k$, $\{\textbf C_\mathcal P\}_{\mathcal P\in\mathcal T}$
\end{algorithmic}
\end{algorithm}

\begin{remark}
Let $r$ denote the maximum hierarchical rank of the $k$th-order, $n$-dimensional dynamic tensor $\mathscr A$ associated with an HTD-based HPDS. The time complexity of a single ALS iteration in Algorithm~\ref{alg:ALS-HTD} can be estimated as $\mathcal O\big(kn^3r^2T+knr^6T\big)$.
\end{remark}
 
We next analyze the convergence properties of Algorithm~\ref{alg:ALS-HTD}. Let $\{\{\textbf V^{\ell}_p\}_{p=1}^k,\{\textbf C_{\mathcal P}^{\ell}\}_{\mathcal P\in\mathcal T}\}_{\ell\ge 0}$ denote the sequence of  factor matrices and transfer matrices generated by Algorithm~\ref{alg:ALS-HTD}.
\begin{proposition}
The sequence of objective function of \eqref{eq:htd-opt}
is monotonically non-increasing and convergent. Moreover, every accumulation point of the iterate sequence
$\{\{\textbf V^{\ell}_p\}_{p=1}^k,\{\textbf C_{\mathcal P}^{\ell}\}_{\mathcal P\in\mathcal T}\}_{\ell\ge 0}$ 
is a first-order stationary point of \eqref{eq:htd-opt}.
\end{proposition}

\textbf{Proof.}
Algorithm~\ref{alg:ALS-HTD} updates blocks in parallel across the dimension tree. Each step exactly solves a convex quadratic least-squares subproblem for its respective factor or transfer matrix, ensuring the objective of \eqref{eq:htd-opt} is non-increasing. Because updates at the same level involve disjoint parameters, they function as independent block coordinate descent steps that preserve monotonicity over each sweep.
Moreover, the objective function is a polynomial in the factor and transfer matrices and is thus semi-algebraic, satisfying the KL property. Hence, every accumulation point of the iterate sequence is a first-order stationary point of \eqref{eq:htd-opt}.
\hfill$\blacksquare$

\begin{proposition}\label{prop:htd2}
Let 
$\{\{\textbf V_p^\star\}_{p=1}^{k},\{\textbf C_{\mathcal P}^\star\}_{\mathcal P\in\mathcal T}\}$ be an isolated stationary point of \eqref{eq:htd-opt}.
Assume that there exists a neighborhood $\mathcal N$
of 
$\{\{\textbf V_p^\star\}_{p=1}^{k},\{\textbf C_{\mathcal P}^\star\}_{\mathcal P\in\mathcal T}\}$
such that, for all iterates in $\mathcal N$, the corresponding
regression matrices $\textbf H_p$ for $p=1,2,\dots,k$ and
$\textbf H_{\mathcal P}$ for $\mathcal P\in\mathcal T$
satisfy
$\textbf H_p^\top \textbf H_p \succeq \alpha_p \textbf I$ and $\textbf H_{\mathcal P}^\top \textbf H_{\mathcal P} \succeq \alpha_{\mathcal P} \textbf I$
for some constants $\alpha_p>0$ and $\alpha_{\mathcal P}>0$.
Then Algorithm~\ref{alg:ALS-HTD} converges locally linearly to
$\{\{\textbf V_p^\star\}_{p=1}^{k},\{\textbf C_{\mathcal P}^\star\}_{\mathcal P\in\mathcal T}\}$.
Specifically, there exists  a neighborhood
$\mathcal N'\subseteq\mathcal N$ such that for all iterates in
$\mathcal N'$, it holds that
\begin{align*}
\sum_{p=1}^k
\|\textbf V_p^{(\ell+1)}-\textbf V_p^\star\|_F^2
+
\sum_{\mathcal P\in\mathcal T}
\|\textbf C_{\mathcal P}^{(\ell+1)}-\textbf C_{\mathcal P}^\star\|_F^2\\
\le
\rho
\left(
\sum_{p=1}^k
\|\textbf V_p^{(\ell)}-\textbf V_p^\star\|_F^2
+
\sum_{\mathcal P\in\mathcal T}
\|\textbf C_{\mathcal P}^{(\ell)}-\textbf C_{\mathcal P}^\star\|_F^2
\right).
\end{align*}
for some $\rho\in(0,1)$.
\end{proposition}

\textbf{Proof.}
Fixing all variables except one block (either $\textbf V_p$ or $\textbf C_{\mathcal P}$), each update solves a strongly convex least-squares problem, hence admits a unique minimizer and defines a well-posed update. The regression matrices depend polynomially on the variables and are therefore locally Lipschitz, so the full ALS sweep defines a locally Lipschitz mapping $\Phi$. Updates within each tree level act on disjoint parameter blocks and can be viewed as independent block-coordinate steps. Under the similar local regularity conditions as in Proposition~\ref{prop:ttd2}, the mapping $\Phi$ is contractive in a neighborhood of the stationary point. Hence, the same local contraction argument applies, yielding local linear convergence.
\hfill$\blacksquare$

Beyond convergence analysis, it is also important to characterize the data informativity conditions that
guarantee identifiability of the HTD-based HPDSs.
\begin{proposition}
A sufficient condition for identification up to the standard gauge freedom of HTD-based HPDSs from data is (\ref{eq:aurankcondition}). A necessary condition  is that the  regression matrices
$\textbf H_p$ and $\textbf H_{\mathcal P}$ 
have full column rank for $p=1,2,\dots,k$ and $\mathcal P\in\mathcal {T}$, respectively.
\end{proposition}

\textbf{Proof.}
Similar to Proposition~\ref{prop:ttdinfo}, the sufficiency follows from the fact that \eqref{eq:aurankcondition} guarantees uniqueness of the full dynamic tensor $\mathscr A$. Since the HTD-based representation parameterizes $\mathscr A$, this implies identifiability of $\{\textbf V_p\}$ and $\{\textbf C_{\mathcal P}\}$ up to the standard gauge freedom. For necessity, under the HTD parameterization, identification reduces to a collection of linear least-squares problems associated with both $\textbf V_p$ and $\textbf C_{\mathcal P}$, characterized by regression matrices $\textbf H_p$ and $\textbf H_{\mathcal P}$, respectively. For each subproblem to admit a unique solution, the corresponding regression matrix must have full column rank. 
\hfill$\blacksquare$

The above results establish the basic theoretical
properties of the proposed HTD-based identification
method, including computational complexity,
convergence guarantees of the ALS iterations, and
identifiability conditions under the HTD
parameterization. The extension to the noisy case follows the similar arguments as in the TTD analysis and is therefore omitted for brevity.


\subsection{Identification of CPD-based HPDSs}
We finally consider the identification of CPD-based HPDSs. In this representation, the dynamic tensor $\mathscr A$ is expressed as a sum of $r$ rank-one tensors, where $r$ denotes the CP-rank. The total number of parameters is $\mathcal O(knr)$, growing linearly with respect to the polynomial order $k$, system dimension $n$, and the CP-rank $r$. Thus, it provides the most compact parameterization. Under the CPD parameterization, the induced polynomial vector field admits a separable structure. For each sampled state $\textbf x(j)$, the system mapping is 
\begin{align}\label{eq: sysCPD}
	\textbf f(j)
=\textbf U^{(k)}
\boldsymbol{\Lambda}
\left(
\textbf x^\top\textbf U^{(k-1)}
\odot
\cdots
\odot
\textbf x^\top\textbf U^{(1)}
\right)^{\top},
	\end{align}
which forms $\textbf F_0$. The identification problem becomes
    \begin{align}\label{eq:cpd-opt}
\min_{\{\textbf U^{(p)}\}_{p=1}^k,\boldsymbol{\Lambda}}  \| \textbf{X}_1 -\textbf F_0 \|_F^2.
\end{align}
Owing to the column-wise scaling invariance of  CPD, the diagonal weight matrix $\boldsymbol{\Lambda}$ can be absorbed into the last factor matrix $\textbf U^{(k)}$ without loss of generality. Defining $\widetilde{\textbf U}^{(k)}= \textbf U^{(k)}\boldsymbol{\Lambda}$,
the problem can be equivalently
formulated in terms of $\{\textbf U^{(p)}\}_{p=1}^{k-1}$ and  $\widetilde{\textbf U}^{(k)}$.  Compared with the TTD and HTD formulations, the CPD-based representation leads to a simpler ALS scheme based on Khatri--Rao products of the remaining factor matrices.

 	The detailed procedure is summarized in Algorithm~\ref{alg:ALS-CPD}. For a $p$th factor matrix
	$\textbf U^{(p)}$, the contributions of all other factor matrices are collectively gathered into the Khatri--Rao product vector
    \begin{align}\label{eq: cpdB}
	\textbf b_p(j)
	=\bigodot_{\substack{q=1, q\neq p}}^{k-1}
	\Big(\textbf x(j)^\top \textbf U^{(q)}\Big),
    \end{align}
    which yields 
\begin{align}\label{eq:cpd-sub}
\min_{\mathrm{vec}(\textbf U^{(p)})}\|
\textbf H_p\,\mathrm{vec}(\textbf U^{(p)})-
\mathrm{vec}(\textbf X_1)\|_2^2.\end{align}
The regression matrix $\textbf H_p$ is formed by stacking 
	\begin{align}\label{eq: cpdH}
	\textbf H_p(j)
	=
	\left(
	\textbf x(j)^\top
	\otimes
	\textbf U^{(k)}\operatorname{diag}\big((\textbf b_p(j)^\top)\big)
	\right)\textbf S_p,
	\end{align}
    where $\textbf S_p$ satisfies  $\mathrm{vec}((\textbf U^{(p})^\top) = \textbf S_p\mathrm{vec}(\textbf U^{(p)}).$
    The least-squares solution is reshaped into $\textbf U^{(p)}\in\mathbb R^{n\times r}$. Column normalization is performed after each update to remove the inherent scaling ambiguity of the CPD. After updating the first $k-1$ factor matrices, the weighted factor matrix $\widetilde{\textbf U}^{(k)}$ is updated by solving an analogous linear least-squares problem.
The diagonal weight matrix $\boldsymbol{\Lambda}$ can be recovered a posteriori from the column scaling of $\widetilde{\textbf U}^{(k)}$ if desired.

    \begin{algorithm}[t]
\caption{ALS Identification for CPD-based HPDS}
\label{alg:ALS-CPD}
\begin{algorithmic}[1]
\State \textbf{Input:} Data $\textbf X_0,\textbf X_1$, tolerance $\epsilon > 0$, truncation threshold $\delta > 0$, maximum iterations $N_{\max}$
\State \textbf{Output:} CP factor matrices $\textbf U^{(1)},\dots,\textbf U^{(k-1)}\in\mathbb R^{n\times r}$ and
$\widetilde{\textbf U}^{(k)}\in\mathbb R^{n\times r}$ (where $\widetilde{\textbf U}^{(k)}=\textbf U^{(k)}\boldsymbol\Lambda$).
\State Randomly initialize $\textbf U^{(1)},\dots,\textbf U^{(k-1)}$ and $\widetilde{\textbf U}^{(k)}$; set $e_{\text{prev}}=+\infty$.
\For{$\text{iter} = 1$ to $N_{\max}$}
    \For{$p = 1$ to $k-1$}
        \For{$j = 1$ to $T$}      
            \State Construct regressor \eqref{eq: cpdH} based on \eqref{eq: cpdB}
        \EndFor
        \State Stack $\textbf H_p = [\textbf H_p(0)^\top\  \cdots\ \textbf H_p(T-1)^\top]^\top$
        \State Solve linear least-squares problem \eqref{eq:cpd-sub}
        \State Reshape solution into $\textbf U^{(p)} \in \mathbb{R}^{n_p\times r_{p}}$
       \State Normalize columns of $\textbf U^{(p)}$
    \EndFor
    \For{$j=1$ to $T$}
        \State Compute
$\textbf b(j)=\bigodot_{\substack{q=1}}^{k-1}
            \big(\textbf x(j))^\top \textbf U^{(q)}\big)$
    
    \EndFor
    \State Form $\textbf B=[\textbf b(1)^\top\ \textbf b(2)^\top\  \cdots\ \textbf b(T)^\top]$.
    \State Update $\widetilde{\textbf U}^{(k)}$ by solving
  $\min_{\widetilde{\textbf U}^{(k)}} \|\widetilde{\textbf U}^{(k)}\textbf B-\textbf X_1\|_F^2.$
    \State Compute $\textbf F_0$ using \eqref{eq: sysCPD} and form the prediction error $e=\|\textbf X_1-\textbf F_0\|_F^2$
    \If{$|e_{\text{prev}} - e|/e_{\text{prev}} < \epsilon$}
        \State \textbf{break}
    \EndIf
    \State $e_{\text{prev}} \gets e$
\EndFor
\State \textbf{return} $\textbf U^{(1)}, \dots,\textbf U^{(k-1)}, \widetilde{\textbf U}^{(k)}$
\end{algorithmic}
\end{algorithm}


\begin{remark}
Let $r$ be the CP-rank of the $k$th-order, $n$-dimensional dynamic tensor $\mathscr A$ associated with a CPD-based HPDS. The computational complexity of a single ALS iteration in Algorithm~\ref{alg:ALS-CPD} can then be estimated as $\mathcal O(k^2nrT + kn^3r^2T)$.
\end{remark}

We next analyze the convergence properties of Algorithm~\ref{alg:ALS-CPD}. Let $\{\{\textbf U^{(p)}_\ell\}_{p=1}^k,\widetilde{\textbf U}_\ell^{(k)}\}_{\ell\ge 0}$ denote the sequence of  factor matrices generated by Algorithm~\ref{alg:ALS-CPD}.

\begin{proposition}
The sequence of objective values of \eqref{eq:cpd-opt}
is monotonically non-increasing and convergent. Furthermore, every accumulation point of the iterate sequence
$\{\{\textbf U_\ell^{(p)}\}_{p=1}^{k-1},
\widetilde{\textbf U}_\ell^{(k)}\}_{\ell\ge 0}$
is a first-order stationary point of \eqref{eq:cpd-opt}.
\end{proposition}
\textbf{Proof.}
Fixing all factor matrices except one block, each update in Algorithm~\ref{alg:ALS-CPD} reduces to a convex quadratic least-squares subproblem, where the regression matrix is formed by the Khatri-Rao product of the remaining factor matrices, leading to a simpler algebraic structure compared to the contraction-based formulations in TTD and HTD. Hence, the objective value of \eqref{eq:cpd-opt} is non-increasing after each update, and the sequence of objective values is monotonically non-increasing and bounded below. The remaining argument follows that of Proposition~\ref{prop:ttd1}: since the objective is a polynomial and thus satisfies the KL property, every accumulation point of the iterate sequence is a first-order stationary point.
\hfill$\blacksquare$

\begin{proposition}\label{prop:cpd2}
Let $\{\{\textbf U_\star^{(p)}\}_{p=1}^{k-1}$,
$\widetilde{\textbf U}_\star^{(k)}\}$
be an isolated stationary point of \eqref{eq:cpd-opt}.
Assume that there exists a neighborhood $\mathcal N$
of this point such that, for all iterates in $\mathcal N$, the corresponding regression
matrix $\textbf H_p$ satisfies
$\textbf H_p^\top \textbf H_p \succeq \alpha_p \textbf I$ for  $p=1,2,\dots,k-1$ and $\textbf B\textbf B^\top \succeq \beta \textbf I$ for some constants $\alpha_p>0$ and $\beta>0$..
Then Algorithm~\ref{alg:ALS-CPD} converges locally
linearly to the stationary point. Specifically, there
exists a neighborhood
$\mathcal N'\subseteq\mathcal N$ such that
\begin{align*}
&\sum_{p=1}^{k-1}
\|\textbf U^{(p)}_{\ell+1}-\textbf U_\star^{(p)}\|_F^2
+
\|\widetilde{\textbf U}^{(k)}_{\ell+1}-\widetilde{\textbf U}_\star^{(k)}\|_F^2,\\
&\le
\rho
\left(
\sum_{p=1}^{k-1}
\|\textbf U^{(p)}_{\ell}-\textbf U_\star^{(p)}\|_F^2
+
\|\widetilde{\textbf U}^{(k)}_{\ell}-\widetilde{\textbf U}_\star^{(k)}\|_F^2
\right).
\end{align*}
for some $\rho\in(0,1)$.
\end{proposition}

\textbf{Proof.}
Each block update solves a strongly convex least-squares problem and therefore admits a unique minimizer. The regression matrices are constructed from Khatri--Rao products of the remaining factor matrices, hence depend polynomially on the variables and are locally Lipschitz. It follows that the full ALS sweep defines a locally Lipschitz mapping. Under the similar local regularity conditions as in Proposition~\ref{prop:ttd2}, the mapping is contractive in a neighborhood of the stationary point, which yields the claimed local linear convergence.
\hfill$\blacksquare$

Finally, we examine the data informativity of the identification problem under the CPD parameterization.

\begin{proposition}
A sufficient condition for identification up to scaling and permutation ambiguities of CPD-based HPDSs from data is given by \eqref{eq:aurankcondition}.
A necessary condition  is that the
 stacked regression matrices
$\textbf H_p$ for $p=1,2,\dots,k-1$ and $\textbf B$
have full column rank.
\end{proposition}
\textbf{Proof.}
For the sufficiency, the CPD-based representation parameterizes $\mathscr A$, and \eqref{eq:aurankcondition} guarantees uniqueness of the full dynamic tensor $\mathscr A$. Hence, the identifiability of the factor matrices up to scaling and permutation ambiguities is implied. For necessity, identification reduces to least-squares subproblems associated with  $\textbf H_p$ and $\textbf B$. For each subproblem to admit a unique solution, these matrices must have full column rank.
\hfill$\blacksquare$

These results provide theoretical guarantees for the
proposed CPD-based identification framework, highlighting its compact parameterization and the
resulting Khatri--Rao regression structure that enables
efficient ALS-based system identification.
The noisy case follows analogously from the TTD analysis and is therefore omitted for brevity.

 \section{Numerical examples}\label{sec: simulation}
	In this section, we evaluate the performance of the proposed TTD-, HTD- and CPD-based identification algorithms following Algorithm~\ref{alg:ALS-TTD}, Algorithm~\ref{alg:ALS-HTD}, and Algorithm~\ref{alg:ALS-CPD} through numerical experiments. The experiments are conducted on synthetically generated HPDSs,
	for which the ground-truth tensor representations are available.
	This allows a quantitative assessment of identification accuracy,
	convergence behavior, and robustness to noise.
	We consider two experimental settings. In the first setting, we generate HPDSs whose dynamic tensors admit
	low-rank structure in a specific format (i.e., TTD, HTD, or CPD), and apply the corresponding ALS identification algorithm
	to assess structure-matched recovery performance. In the second setting, we consider HPDSs with general dynamic tensors,
	and apply identification algorithms to the same data to compare the reconstruction accuracy, robustness and computation cost. All simulations were carried out in MATLAB R2023b on a machine equipped with an Apple M3 chip and 16~GB of RAM. Tensor operations and tensor factorizations were implemented using Tensor Toolbox~3.6~\cite{kolda2023matlab}. The associated code can be found in https://github.com/XinMao0/tensor-decomposition-based-HPDSs-identification.
	
	\subsection{Convergence Analysis }
	In this experiment, we investigate the convergence behavior and identification accuracy of the proposed algorithms. The dynamic tensor $\mathscr A$ is generated using three different low-rank generation schemes: low TT-rank generation, low HT-rank generation, and low CP-rank generation. We set the system dimension $n=9$ and polynomial order $k=4$. The TT-ranks are set to 
	$ \{r_0,r_1,r_2,r_3,r_4,r_5\}=\{ 1,9,10,3,1 \}.$
	We use a balanced binary dimension tree 
	with node sets 
	$\{\{1,2,3,4\},\{1,2\},\{3,4\},
	\{1\},\{2\},\{3\},\{4\}\}$. 
	The hierarchical ranks associated with these nodes are
	$\{1,10,10,9,9,9,3\}$, respectively. The CP-rank is set to $r=3$. State trajectories are generated by numerically integrating
	\eqref{eq: ausys} from random initial conditions. The state and its time derivative are sampled at $T$ time instants with step size $\tau$ to form the data matrices $\textbf X_0$ and $\textbf X_1$. For each generated system, we apply the corresponding identification algorithm to the collected data set. The prediction and identification errors are defined as
\begin{equation*}
E_{\mathrm{pred}}
=
\frac{\|\mathbf X_1 - \widehat{\mathbf X}_1\|_F}{\|\mathbf X_1\|_F},\
E_{\mathscr A}
=
\frac{\|\mathscr A - \widehat{\mathscr A}\|_F}{\|\mathscr A\|_F},
\end{equation*}
respectively.
	Fig.~\ref{fig: convergence} illustrates the convergence behavior of the ALS iterations and the accuracy of the recovered tensor factors for the three low-rank generation schemes. The results demonstrate that all three methods converge reliably and accurately recover the underlying tensor factors when the assumed low-rank structure matches the true generative model.

	\begin{figure}[t]
		\centering
		\includegraphics[width=\linewidth]{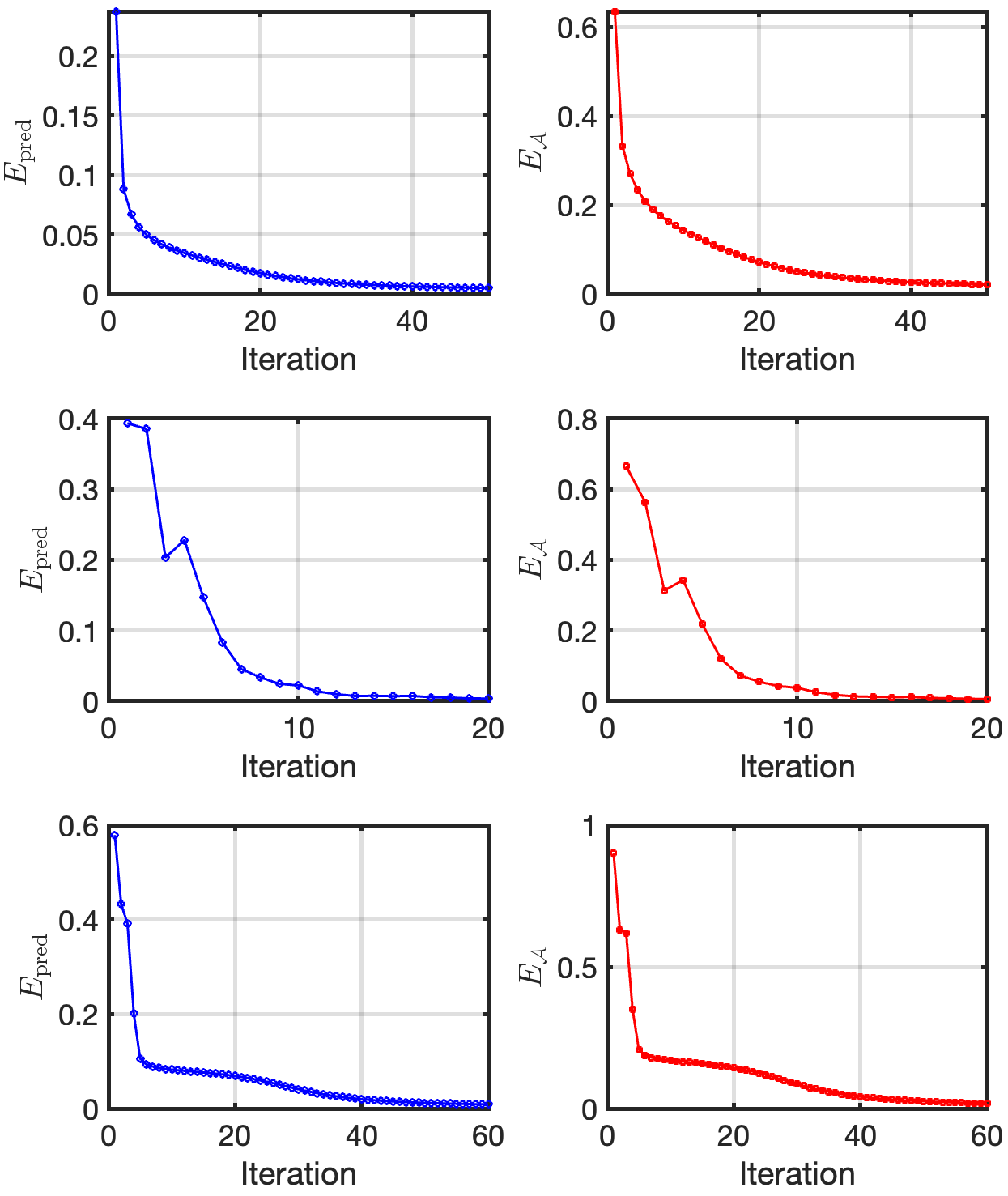}
		\caption{Convergence of the TTD-, HTD-, and CPD-based ALS HPDS identification algorithms (Algorithms~\ref{alg:ALS-TTD},\ref{alg:ALS-HTD}, and \ref{alg:ALS-CPD}). 
			Rows correspond to TTD, HTD, and CPD (top to bottom), and columns show the relative prediction and identification errors, respectively. 
			The errors decrease steadily across sweeps, reflecting the block coordinate-descent nature of the algorithms.}
		\label{fig: convergence}
	\end{figure}
	
	
	\subsection{Accuracy and Noise Robustness}
	In this experiment, we evaluate the performance of the proposed
	identification methods on general HPDSs whose coefficient
	tensors do not admit any particular low-rank tensor structure. We consider HPDSs of polynomial order $k=4$ and dimension $n=9$.
	The coefficient tensor
	$\mathscr A$ is generated as a random sparse tensor without imposing any TTD-, HTD-, or CPD-based structure. Specifically,  the nonzero entries of $\mathscr A$ are independently drawn from a standard Gaussian distribution. The sparsity level is chosen to be $0.001$. State trajectories are generated by numerically integrating
	\eqref{eq: ausys} from randomly sampled initial conditions.
	The state and its time derivative are sampled at $T$ time instants
	such that the lifted data matrix $\hat{\textbf X}_0$
	satisfies the rank condition required for identifiability.
	All identification methods are applied to the same data sets
	for a fair comparison.

	We compare the lifting-based identification method in \cite{mao2025tensor}, which reconstructs the full tensor in the ambient space, with the proposed TTD-, HTD-, and CPD-based ALS algorithms that exploit low-rank tensor structures. The relative tensor identification errors are reported in Table~\ref{tab:reconstruction_comparison}.
	with the proposed TTD-, HTD-, and CPD-based ALS algorithms
	in terms of relative tensor identification error, as reported in Table~\ref{tab:reconstruction_comparison}.
	Although the true tensor does not conform to a specific low-rank format, the ALS-based methods consistently achieve smaller reconstruction errors
	than the lifting-based approach.
	This indicates that structured low-rank parameterizations
	can provide improved numerical stability and regularization effects even when the ground truth is not exactly low-rank.
	To further assess robustness, we perturb the sampled data
	with additive Gaussian noise of varying levels. Fig.~\ref{fig: noiseconvergence} illustrates the relative identification errors under increasing noise intensities.
	The proposed ALS-based methods exhibit improved robustness
	compared with the lifting-based scheme, maintaining lower identification errors across all noise levels.

		\begin{table}[t]
		\centering
		\caption{Relative identification errors of the lifting-based method in \cite{mao2025tensor} and the proposed low-rank identification schemes for general sparse HPDS.}
		\label{tab:reconstruction_comparison}
		\centering
		\begin{tabular}{c|c|c}
			\hline
			\textbf{Tensor Format} & \textbf{Lifting} & \textbf{Proposed} \\
			\hline
			TTD  & $4.275\times 10^{-3}$ 
			&$5.06 \times 10^{-4}$ \\
			HTD  & $7.405 \times 10^{-3}$ 
			& $6.05 \times 10^{-4}$ \\
			CPD  & $6.736\times 10^{-3}$  &$7.405 \times 10^{-4}$    \\
			\hline
		\end{tabular}
	\end{table}

	\begin{figure}[t]
		\centering
		\includegraphics[width=\linewidth]{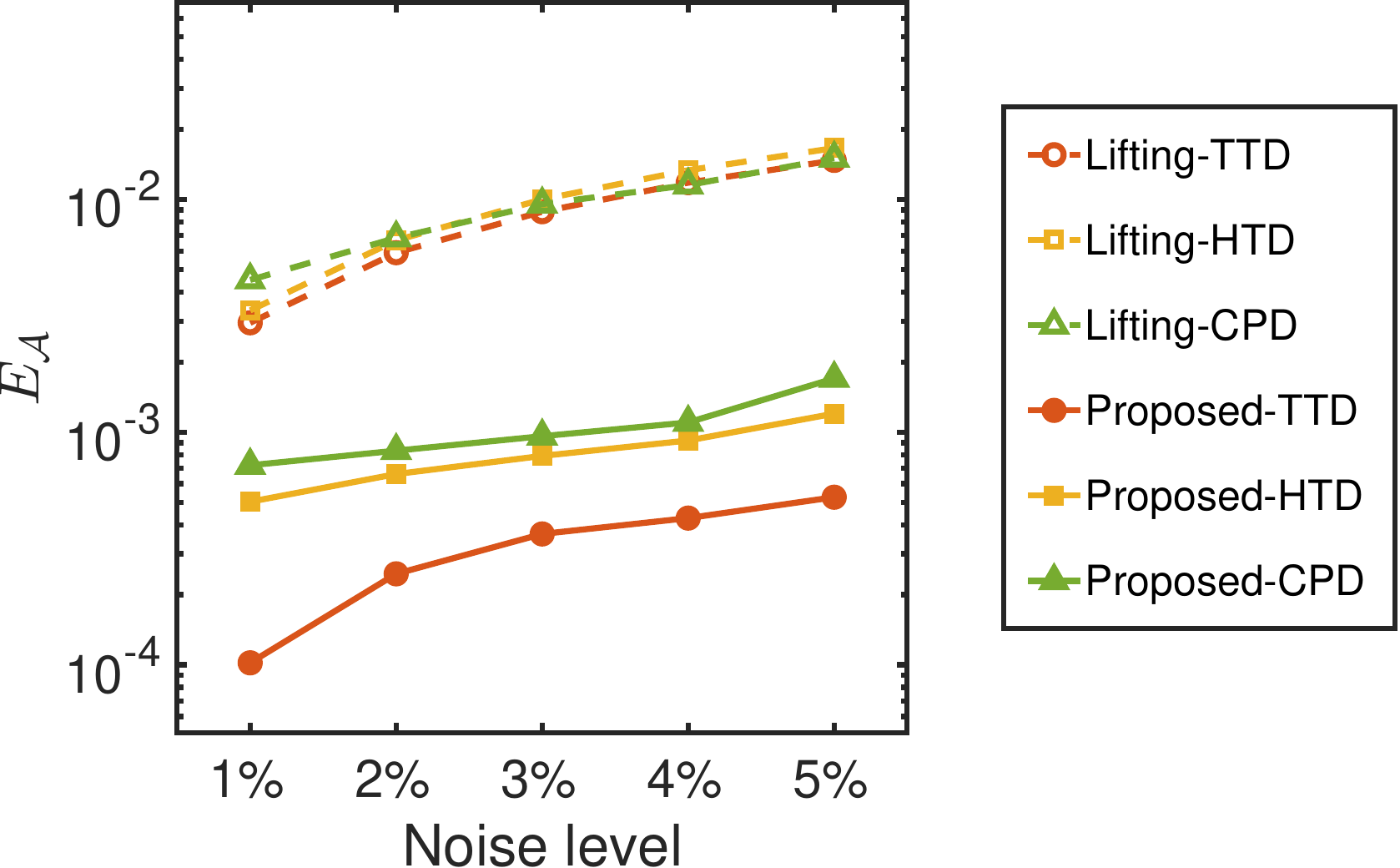}
		\caption{Relative identification errors of the lifting-based  methods and proposed methods  for TTD-, HTD-, and CPD-based HPDSs under increasing measurement noise levels.}
		\label{fig: noiseconvergence}
	\end{figure}

	\subsection{Computational Scalability}
	In this experiment, we investigate the computational scalability of the proposed identification methods with respect to the system dimension. Specifically, we examine how the computational time grows as the state dimension $n$ increases. We consider HPDSs of  order $k=7$ and vary the system dimension $n$ 
    over a range from $8$ to $400$. The data generation procedure follows the same setup as in the previous experiment, and all simulations are performed under identical stopping criteria. 
Fig.~\ref{fig: time} reports the computational time with respect to the system dimension $n$. In the small-scale regime as shown in Fig.~\ref{fig: time}(a), both the lifting-based approaches and the proposed ALS-based methods are executable, allowing for a direct comparison. However, as $n$ increases, the lifting-based methods exhibit rapid growth in computational cost and encounter out-of-memory (OOM) issues.
In contrast, the proposed ALS-based methods remain computationally tractable and exhibit moderate growth with respect to $n$, as shown in Fig.~\ref{fig: time}(b). This behavior is consistent with the theoretical complexity analysis, which predicts polynomial scaling in $n$ for tensor-structured ALS updates. These results highlight the computational advantage of the proposed tensor-structured identification framework, enabling efficient learning for medium- and large-scale polynomial dynamical systems.

	\begin{figure}[t]
		\centering
	\includegraphics[width=\linewidth]{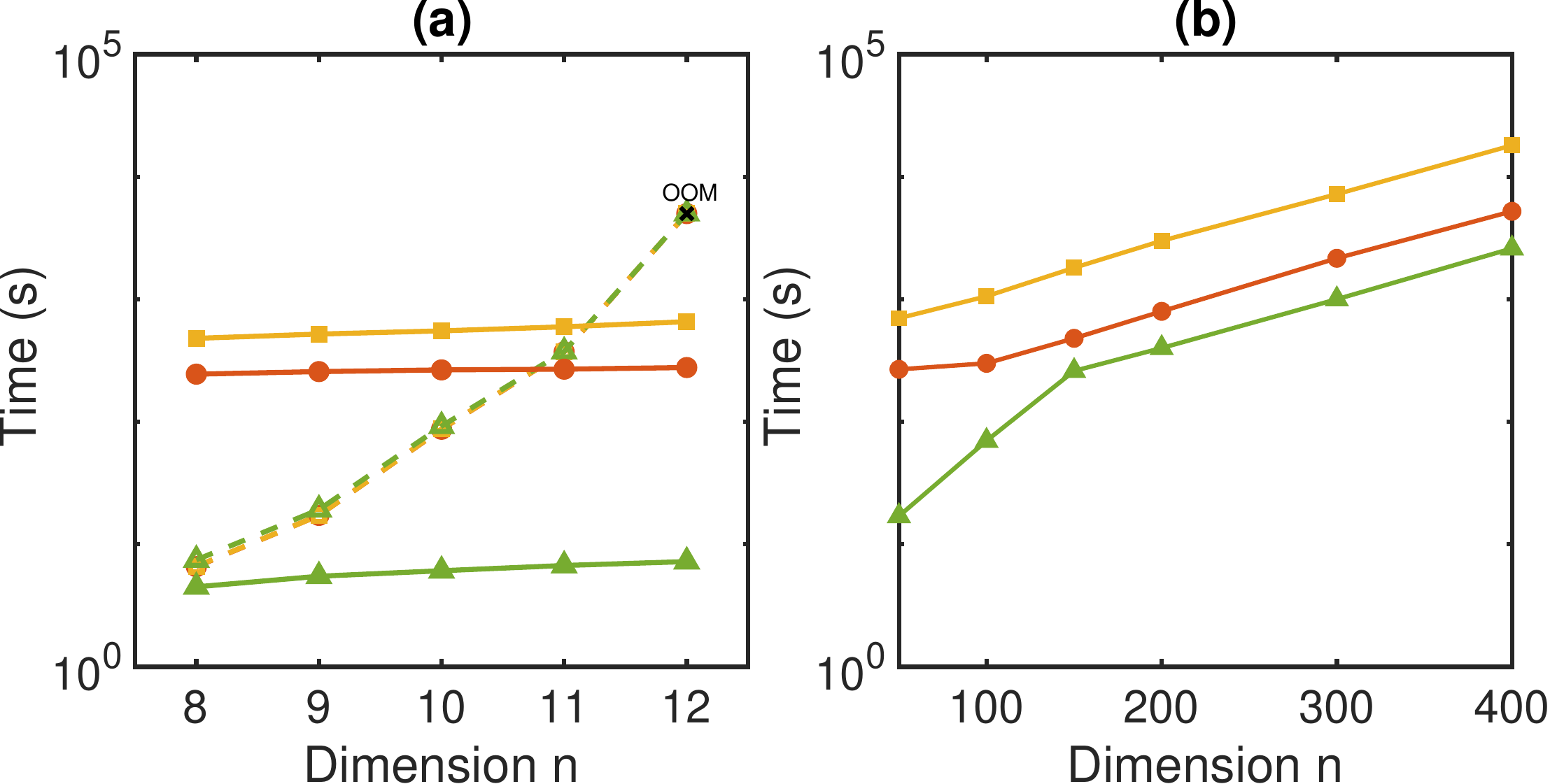}
		\caption{
			Computation time comparison between the lifting-based approaches and the proposed ALS-based methods for TTD, HTD, and CPD models with $k=7$.	Line styles and markers are consistent with Fig.~\ref{fig: noiseconvergence}.
			(a) Small-scale regime with varying state dimension $n$ from $8$ to $12$, where both approaches are applicable. 
			The lifting-based methods encounter OOM issues at $n=12$.	
			(b) Scalability of the proposed methods for larger $n$, where only the ALS-based approaches remain computationally feasible.}
		\label{fig: time}
	\end{figure}

\section{Conclusion}\label{sec: conclusion}
In this article, we present a scalable  framework for the identification of HPDSs based on structured low-rank tensor decompositions. By leveraging TTD-, HTD-, and CPD-based representations, the proposed approach learns compact tensor or matrix factors directly from data, thereby avoiding the construction of the full higher-order dynamic tensor. The proposed ALS-based identification algorithms exploit the underlying multilinear structure of the tensor models, leading to efficient parameter estimation with significantly reduced computational and memory complexity. Both theoretical analysis and numerical experiments demonstrate that the proposed methods achieve accurate and robust identification, while maintaining favorable scalability with respect to the system dimension. In particular, the results highlight that tensor-structured representations provide a principled way to mitigate the curse of dimensionality inherent in higher-order polynomial systems. Beyond computational efficiency, the proposed framework offers a flexible modeling paradigm for high-dimensional nonlinear systems with structured and higher-order interactions. The proposed framework is also applicable to discrete-time systems, where similar tensor-structured representations and ALS-based identification procedures can be employed.

Future work will pursue several directions. One important direction is to establish sharper identifiability conditions under low-rank tensor parameterizations, in particular by characterizing necessary and sufficient conditions that explicitly depend on the tensor ranks and data richness. 
Another key aspect is to develop a deeper theoretical understanding of the effects of model mismatch, approximation error, and stochastic noise on the recovery performance, including finite-sample guarantees and robustness bounds. 
From an algorithmic perspective, further improvements will be explored to enhance scalability and reliability, such as adaptive rank selection, initialization strategies, and acceleration schemes for ALS-type methods. 
Moreover, extending the framework to partially observed and noisy measurement settings, including missing data, discrete dynamics, and irregular sampling, is of practical importance. The proposed framework has potential applications in real-world large-scale systems, including networked dynamical systems, hypergraph-based interaction models, and nonlinear systems arising in control, robotics, and biological and ecological applications. 
In particular, the ability to capture higher-order interactions makes the approach promising for modeling complex multi-agent and multi-body dynamics in data-driven settings.

 \bibliographystyle{IEEEtran}
\bibliography{reference}

@book{ljung1998system,
	title={System identification},
	author={Ljung, Lennart},
	booktitle={Signal analysis and prediction},
	pages={163--173},
	year={1998},
	publisher={Springer}
}

@article{pickard2024geometric,
  title={Geometric aspects of observability of hypergraphs},
  author={Pickard, Joshua and Stansbury, Cooper and Surana, Amit and Rajapakse, Indika and Bloch, Anthony},
  journal={Ifac-papersonline},
  volume={58},
  number={6},
  pages={321--326},
  year={2024},
  publisher={Elsevier}
}

@article{pickard2024kronecker,
  title={Kronecker product of tensors and hypergraphs: structure and dynamics},
  author={Pickard, Joshua and Chen, Can and Stansbury, Cooper and Surana, Amit and Bloch, Anthony M and Rajapakse, Indika},
  journal={SIAM Journal on Matrix Analysis and Applications},
  volume={45},
  number={3},
  pages={1621--1642},
  year={2024},
  publisher={SIAM}
}

@book{van2012subspace,
  title={Subspace Identification for Linear Systems: Theory—Implementation—Applications},
  author={Van Overschee, Peter and De Moor, BL0888},
  year={2012},
  publisher={Springer Science \& Business Media}
}

@article{brockett1976volterra,
  title={Volterra series and geometric control theory},
  author={Brockett, Roger W},
  journal={Automatica},
  volume={12},
  number={2},
  pages={167--176},
  year={1976},
  publisher={Elsevier}
}

@article{guo2020constructing,
  title={Constructing least-squares polynomial approximations},
  author={Guo, Ling and Narayan, Akil and Zhou, Tao},
  journal={SIAM Review},
  volume={62},
  number={2},
  pages={483--508},
  year={2020},
  publisher={SIAM}
}

@article{phan2020stable,
  title={Stable low-rank tensor decomposition for compression of convolutional neural network},
  author={Phan, Anh-Huy and Sobolev, Konstantin and Sozykin, Konstantin and Ermilov, Dmitry and Gusak, Julia and Tichavsk{\`y}, Petr and Glukhov, Valeriy and Oseledets, Ivan and Cichocki, Andrzej},
 journal={European Conference on Computer Vision},
  pages={522--539},
  year={2020},
  organization={Springer}
}

@article{yan2014image,
  title={Image-based process monitoring using low-rank tensor decomposition},
  author={Yan, Hao and Paynabar, Kamran and Shi, Jianjun},
  journal={IEEE Transactions on Automation Science and Engineering},
  volume={12},
  number={1},
  pages={216--227},
  year={2014},
  publisher={IEEE}
}

@article{pfluger2010spatially,
  title={Spatially adaptive sparse grids for high-dimensional data-driven problems},
  author={Pfl{\"u}ger, Dirk and Peherstorfer, Benjamin and Bungartz, Hans-Joachim},
  journal={Journal of Complexity},
  volume={26},
  number={5},
  pages={508--522},
  year={2010},
  publisher={Elsevier}
}

@article{bertsimas2020sparsehigh,
  title={Sparse high-dimensional regression},
  author={Bertsimas, Dimitris and Van Parys, Bart},
  journal={The Annals of Statistics},
  volume={48},
  number={1},
  pages={300--323},
  year={2020},
  publisher={JSTOR}
}

@article{jansson2003subspace,
  title={Subspace identification and ARX modeling},
  author={Jansson, Magnus},
  journal={IFAC Proceedings Volumes},
  volume={36},
  number={16},
  pages={1585--1590},
  year={2003},
  publisher={Elsevier}
}

@article{de1997stabilizing,
  title={Stabilizing predictive control of nonlinear ARX models},
  author={De Nicolao, Giuseppe and Magni, Lalo and Scattolini, Riccardo},
  journal={Automatica},
  volume={33},
  number={9},
  pages={1691--1697},
  year={1997},
  publisher={Elsevier}
}

@article{watson1967linear,
  title={Linear least squares regression},
  author={Watson, Geoffrey S},
  journal={The Annals of Mathematical Statistics},
  pages={1679--1699},
  year={1967},
  publisher={JSTOR}
}

@article{chen2022explicit,
	title={Explicit solutions and stability properties of homogeneous polynomial dynamical systems},
	author={Chen, Can},
	journal={IEEE Transactions on Automatic Control},
	volume={68},
	number={8},
	pages={4962--4969},
	year={2022},
	publisher={IEEE}
}

@article{batselier2022low,
	title={Low-rank tensor decompositions for nonlinear system identification: A tutorial with examples},
	author={Batselier, Kim},
	journal={IEEE Control Systems Magazine},
	volume={42},
	number={1},
	pages={54--74},
	year={2022},
	publisher={IEEE}
}

@article{kargas2020nonlinear,
  title={Nonlinear system identification via tensor completion},
  author={Kargas, Nikos and Sidiropoulos, Nicholas D},
 journal={Proceedings of the AAAI Conference on Artificial Intelligence},
  volume={34},
  number={04},
  pages={4420--4427},
  year={2020}
}

@article{bousse2017tensor,
  title={Tensor-based large-scale blind system identification using segmentation},
  author={Bouss{\'e}, Martijn and Debals, Otto and De Lathauwer, Lieven},
  journal={IEEE Transactions on Signal Processing},
  volume={65},
  number={21},
  pages={5770--5784},
  year={2017},
  publisher={IEEE}
}

@article{batselier2017tensor,
	title={A Tensor Network Kalman filter with an application in recursive MIMO Volterra system identification},
	author={Batselier, Kim and Chen, Zhongming and Wong, Ngai},
	journal={Automatica},
	volume={84},
	pages={17--25},
	year={2017},
	publisher={Elsevier}
}

@article{brunton2016discovering,
  title={Discovering governing equations from data by sparse identification of nonlinear dynamical systems},
  author={Brunton, Steven L and Proctor, Joshua L and Kutz, J Nathan},
  journal={Proceedings of the national academy of sciences},
  volume={113},
  number={15},
  pages={3932--3937},
  year={2016},
  publisher={National Academy of Sciences}
}

@article{cui2024discrete,
  title={On discrete-time polynomial dynamical systems on hypergraphs},
  author={Cui, Shaoxuan and Zhang, Guofeng and Jard{\'o}n-Kojakhmetov, Hildeberto and Cao, Ming},
  journal={IEEE Control Systems Letters},
  volume={8},
  pages={1078--1083},
  year={2024},
  publisher={IEEE}
}

@article{kolda2009tensor,
  title={Tensor decompositions and applications},
  author={Kolda, Tamara G and Bader, Brett W},
  journal={SIAM review},
  volume={51},
  number={3},
  pages={455--500},
  year={2009},
  publisher={SIAM}
}

@article{craciun2019polynomial,
  title={Polynomial dynamical systems, reaction networks, and toric differential inclusions},
  author={Craciun, Gheorghe},
  journal={SIAM Journal on Applied Algebra and Geometry},
  volume={3},
  number={1},
  pages={87--106},
  year={2019},
  publisher={SIAM}
}

@book{brauer2019mathematical,
  title={Mathematical models in epidemiology},
  author={Brauer, Fred and Castillo-Chavez, Carlos and Feng, Zhilan and others},
  volume={32},
  year={2019},
  publisher={Springer}
}

@article{gorodetsky2018high,
  title={High-dimensional stochastic optimal control using continuous tensor decompositions},
  author={Gorodetsky, Alex and Karaman, Sertac and Marzouk, Youssef},
  journal={The International Journal of Robotics Research},
  volume={37},
  number={2-3},
  pages={340--377},
  year={2018},
  publisher={SAGE Publications Sage UK: London, England}
}

@article{oseledets2011tensor,
	title={Tensor-train decomposition},
	author={Oseledets, Ivan V},
	journal={SIAM Journal on Scientific Computing},
	volume={33},
	number={5},
	pages={2295--2317},
	year={2011},
	publisher={SIAM}
}

@article{oseledets2009breaking,
	title={Breaking the curse of dimensionality, or how to use SVD in many dimensions},
	author={Oseledets, Ivan V and Tyrtyshnikov, Eugene E},
	journal={SIAM Journal on Scientific Computing},
	volume={31},
	number={5},
	pages={3744--3759},
	year={2009},
	publisher={SIAM}
}

@book{guckenheimer2013nonlinear,
	title={Nonlinear {O}scillations, {D}ynamical {S}ystems, and {B}ifurcations of {V}ector {F}ields},
	author={Guckenheimer, John and Holmes, Philip},
	volume={42},
	year={2013},
	publisher={Springer Science \& Business Media}
}

@article{grasedyck2010hierarchical,
	title={Hierarchical singular value decomposition of tensors},
	author={Grasedyck, Lars},
	journal={SIAM journal on matrix analysis and applications},
	volume={31},
	number={4},
	pages={2029--2054},
	year={2010},
	publisher={SIAM}
}

@article{kolda2023matlab,
	title={Tensor Toolbox for MATLAB, Version 3.6},
	author={Bader, Brett W and Kolda, Tamara G},
        journal={www.tensortoolbox.org},
	year={2023},
}

@article{phan2013candecomp,
  title={CANDECOMP/PARAFAC decomposition of high-order tensors through tensor reshaping},
  author={Phan, Anh-Huy and Tichavsk{\`y}, Petr and Cichocki, Andrzej},
  journal={IEEE Transactions on Signal Processing},
  volume={61},
  number={19},
  pages={4847--4860},
  year={2013},
  publisher={IEEE}
}

@article{jansson1998consistency,
  title={On consistency of subspace methods for system identification},
  author={Jansson, Magnus and Wahlberg, Bo},
  journal={Automatica},
  volume={34},
  number={12},
  pages={1507--1519},
  year={1998},
  publisher={Elsevier}
}

@article{hong2020generalized,
  title={Generalized canonical polyadic tensor decomposition},
  author={Hong, David and Kolda, Tamara G and Duersch, Jed A},
  journal={SIAM review},
  volume={62},
  number={1},
  pages={133--163},
  year={2020},
  publisher={SIAM}
}

@article{delabays2025hypergraph,
  title={Hypergraph reconstruction from dynamics},
  author={Delabays, Robin and De Pasquale, Giulia and D{\"o}rfler, Florian and Zhang, Yuanzhao},
  journal={Nature Communications},
  volume={16},
  number={1},
  pages={2691},
  year={2025},
  publisher={Nature Publishing Group UK London}
}

@article{lubich2013dynamical,
  title={Dynamical approximation by hierarchical Tucker and tensor-train tensors},
  author={Lubich, Christian and Rohwedder, Thorsten and Schneider, Reinhold and Vandereycken, Bart},
  journal={SIAM Journal on Matrix Analysis and Applications},
  volume={34},
  number={2},
  pages={470--494},
  year={2013},
  publisher={SIAM}
}

@article{bertsimas2020sparse,
  title={Sparse high-dimensional regression},
  author={Bertsimas, Dimitris and Van Parys, Bart},
  journal={The Annals of Statistics},
  volume={48},
  number={1},
  pages={300--323},
  year={2020},
  publisher={JSTOR}
}

@article{ostertagova2012modelling,
  title={Modelling using polynomial regression},
  author={Ostertagov{\'a}, Eva},
  journal={Procedia engineering},
  volume={48},
  pages={500--506},
  year={2012},
  publisher={Elsevier}
}

@article{mao2025model,
  title={Model Reduction of Homogeneous Polynomial Dynamical Systems via Tensor Decomposition},
  author={Mao, Xin and Chen, Can},
  journal={IEEE Transactions on Automatic Control},
  year={2026},
  publisher={IEEE}
}

@article{mao2025identification,
  title={Identification of Hypergraph Dynamics via Physics-Informed Neural Networks},
  author={Mao, Xin and Dong, Anqi and Chen, Can},
  journal={IEEE Control Systems Letters},
  volume={9},
  pages={2525--2530},
  year={2025},
  publisher={IEEE}
}

@article{cui2025analysis,
  title={Analysis of higher-order Lotka-Volterra models: Application of S-tensors and the polynomial complementarity problem},
  author={Cui, Shaoxuan and Zhao, Qi and Zhang, Guofeng and Jardon-Kojakhmetov, Hildeberto and Cao, Ming},
  journal={IEEE Transactions on Automatic Control},
  year={2025},
  publisher={IEEE}
}

@article{proctor2016dynamic,
  title={Dynamic mode decomposition with control},
  author={Proctor, Joshua L and Brunton, Steven L and Kutz, J Nathan},
  journal={SIAM Journal on Applied Dynamical Systems},
  volume={15},
  number={1},
  pages={142--161},
  year={2016},
  publisher={SIAM}
}

@article{zhao2010recursive,
  title={Recursive identification for nonlinear ARX systems based on stochastic approximation algorithm},
  author={Zhao, Wen-Xiao and Chen, Han-Fu and Zheng, Wei Xing},
  journal={IEEE Transactions on Automatic Control},
  volume={55},
  number={6},
  pages={1287--1299},
  year={2010},
  publisher={IEEE}
}

@book{chen2024tensor,
  title={Tensor-based dynamical systems: {T}heory and {A}pplications},
  author={Chen, Can},
  year={2024},
  publisher={Springer Nature}
}

@book{chesi2009homogeneous,
  title={Homogeneous polynomial forms for robustness analysis of uncertain systems},
  author={Chesi, Graziano and Garulli, Andrea and Tesi, Alberto and Vicino, Antonio},
  volume={390},
  year={2009},
  publisher={Springer Science \& Business Media}
}

@article{samardzija1983stability,
  title={Stability properties of autonomous homogeneous polynomial differential systems},
  author={Samardzija, Nikola},
  journal={Journal of Differential Equations},
  volume={48},
  number={1},
  pages={60--70},
  year={1983},
  publisher={Elsevier}
}

@article{isaksson2002identification,
  title={Identification of ARX-models subject to missing data},
  author={Isaksson, Alf J},
  journal={IEEE Transactions on Automatic Control},
  volume={38},
  number={5},
  pages={813--819},
  year={2002},
  publisher={IEEE}
}

@article{najm2009uncertainty,
  title={Uncertainty quantification and polynomial chaos techniques in computational fluid dynamics},
  author={Najm, Habib N},
  journal={Annual review of fluid mechanics},
  volume={41},
  number={1},
  pages={35--52},
  year={2009},
  publisher={Annual Reviews}
}

@article{ivakhnenko2007polynomial,
  title={Polynomial theory of complex systems},
  author={Ivakhnenko, Alexey Grigorevich},
  journal={IEEE transactions on Systems, Man, and Cybernetics},
  number={4},
  pages={364--378},
  year={2007},
  publisher={IEEE}
}

@book{cox2020applications,
  title={Applications of polynomial systems},
  author={Cox, David A},
  volume={134},
  year={2020},
  publisher={American Mathematical Soc.}
}

@article{carlberg2013gnat,
  title={The GNAT method for nonlinear model reduction: effective implementation and application to computational fluid dynamics and turbulent flows},
  author={Carlberg, Kevin and Farhat, Charbel and Cortial, Julien and Amsallem, David},
  journal={Journal of Computational Physics},
  volume={242},
  pages={623--647},
  year={2013},
  publisher={Elsevier}
}

@article{lassila2014model,
  title={Model order reduction in fluid dynamics: challenges and perspectives},
  author={Lassila, Toni and Manzoni, Andrea and Quarteroni, Alfio and Rozza, Gianluigi},
  journal={Reduced Order Methods for modeling and computational reduction},
  pages={235--273},
  year={2014},
  publisher={Springer}
}

@article{grilli2017higher,
  title={Higher-order interactions stabilize dynamics in competitive network models},
  author={Grilli, Jacopo and Barab{\'a}s, Gy{\"o}rgy and Michalska-Smith, Matthew J and Allesina, Stefano},
  journal={Nature},
  volume={548},
  number={7666},
  pages={210--213},
  year={2017},
  publisher={Nature Publishing Group UK London}
}

@article{cencetti2021temporal,
  title={Temporal properties of higher-order interactions in social networks},
  author={Cencetti, Giulia and Battiston, Federico and Lepri, Bruno and Karsai, M{\'a}rton},
  journal={Scientific reports},
  volume={11},
  number={1},
  pages={7028},
  year={2021},
  publisher={Nature Publishing Group UK London}
}

@article{malizia2024reconstructing,
  title={Reconstructing higher-order interactions in coupled dynamical systems},
  author={Malizia, Federico and Corso, Alessandra and Gambuzza, Lucia Valentina and Russo, Giovanni and Latora, Vito and Frasca, Mattia},
  journal={Nature Communications},
  volume={15},
  number={1},
  pages={5184},
  year={2024},
  publisher={Nature Publishing Group UK London}
}

@article{ragnarsson2012block,
  title={Block tensor unfoldings},
  author={Ragnarsson, Stefan and Van Loan, Charles F},
  journal={SIAM Journal on Matrix Analysis and Applications},
  volume={33},
  number={1},
  pages={149--169},
  year={2012},
  publisher={SIAM}
}

@article{mao2025tensor,
  title={Tensor-based homogeneous polynomial dynamical system analysis from data},
  author={Mao, Xin and Dong, Anqi and He, Ziqin and Mei, Yidan and Mei, Shenghan and Chen, Can},
  journal={arXiv preprint arXiv:2503.17774},
  year={2025}
}

@article{pickard2023observability,
	title={Observability of hypergraphs},
	author={Pickard, Joshua and Surana, Amit and Bloch, Anthony and Rajapakse, Indika},
	journal={2023 62nd IEEE Conference on Decision and Control (CDC)},
	pages={2445--2451},
	year={2023},
	organization={IEEE}
}

@article{bairey2016high,
  title={High-order species interactions shape ecosystem diversity},
  author={Bairey, Eyal and Kelsic, Eric D and Kishony, Roy},
  journal={Nature communications},
  volume={7},
  number={1},
  pages={12285},
  year={2016},
  publisher={Nature Publishing Group UK London}
}

@article{chellaboina2009modeling,
  title={Modeling and analysis of mass-action kinetics},
  author={Chellaboina, Vijaysekhar and Bhat, Sanjay P and Haddad, Wassim M and Bernstein, Dennis S},
  journal={IEEE Control Systems Magazine},
  volume={29},
  number={4},
  pages={60--78},
  year={2009},
  publisher={IEEE}
}

@article{chen2021controllability,
  title={Controllability of hypergraphs},
  author={Chen, Can and Surana, Amit and Bloch, Anthony M and Rajapakse, Indika},
  journal={IEEE Transactions on Network Science and Engineering},
  volume={8},
  number={2},
  pages={1646--1657},
  year={2021},
  publisher={IEEE}
}

@article{dong2024controllability,
  title={Controllability and Observability of Temporal Hypergraphs},
  author={Dong, Anqi and Mao, Xin and Vasudevan, Ram and Chen, Can},
  journal={IEEE Control Systems Letters},
  year={2024},
  publisher={IEEE}
}

\newpage
\appendix
\section{Algorithms}


\noindent
{\setlength{\arrayrulewidth}{0.8pt}
\begin{tabular*}{\linewidth}{@{\extracolsep{\fill}}l}
\hline
\textbf{Function} \textsc{LeafLS}
$(p,\textbf X_0,\{\textbf V_q\},\{\textbf C_{\mathcal Q}\},\mathcal T)$ \\
\hline
\end{tabular*}}
\begin{algorithmic}[1]
\State \textbf{Input:} leaf index $p$, data $\textbf X_0=[\textbf x(1)\ \cdots\ \textbf x(T)]$, factors $\{\textbf V_q\}_{q=1}^k$, transfers $\{\textbf C_\mathcal Q\}_{\mathcal Q\in\mathcal T}$, tree $\mathcal T$
\State \textbf{Output:} stacked regressor matrix $\textbf H_p$
\For{$j=1$ to $T$}
    \State Construct left contraction \eqref{eq:htdleft}
    \State Construct right contraction
    \[
    \textbf R_p(j)
    =
    (\otimes_{\mathcal Q\in\mathcal G_{d-1}}\textbf C_\mathcal Q)\cdots
    (\otimes_{\mathcal Q\in\mathcal G_{0}}\textbf C_\mathcal Q)
    \]
    \State Construct $\textbf H_p(j)$ based on \eqref{eq:htdH}
\EndFor
\State Stack $\textbf H_p=[\textbf H_p(1)^\top\ \textbf H_p(2)^\top\ \cdots\ \textbf H_p(T)^\top]^\top$
\State \Return $\textbf H_p$
\end{algorithmic}
\vspace{-17pt}
\rule{\linewidth}{0.6pt}

\noindent
{\setlength{\arrayrulewidth}{0.6pt}
\begin{tabular*}{\linewidth}{@{\extracolsep{\fill}}l}
\hline
\textbf{Function} \textsc{InternalLS}
$(\mathcal P,l,\textbf X_0,\{\textbf V_q\},\{\textbf C_\mathcal Q\},\mathcal T)$\\
\hline
\end{tabular*}}
\begin{algorithmic}[1]
\State \textbf{Input:} internal node $\mathcal P\in\mathcal G_l$, level $l$, data $\textbf X_0$, factors $\{\textbf V_q\}_{q=1}^k$, transfers $\{\textbf C_\mathcal Q\}_{\mathcal Q\in\mathcal T}$, tree $\mathcal T$
\State \textbf{Output:} stacked regressor matrix $\textbf H_\mathcal P$
\For{$j=1$ to $T$}
    \State Compute right contraction with $\textbf C_{\mathcal P}$ replaced by $\textbf I_{r_{\mathcal P}}$
    \[
    \textbf R_{\mathcal P}(j)
    =
    (\otimes_{\mathcal Q\in\mathcal G_{l}}\textbf C_{\mathcal Q})
    \cdots
    (\otimes_{\mathcal Q\in\mathcal G_{0}}\textbf C_{\mathcal Q})
    \]
     \State Compute left contraction
            \begin{align*}
            \textbf L_{\mathcal P}(j)
            =
            \left(
            \textbf V_k
            \otimes
            \textbf x(j)^\top\textbf V_{k-1}
            \otimes \cdots \otimes
            \textbf x(j)^\top\textbf V_{1}
            \right)\\
            (\otimes_{\mathcal Q\in\mathcal G_{d-1}}\textbf C_{\mathcal Q})
            \cdots
            (\otimes_{\mathcal Q\in\mathcal G_{l+1}}\textbf C_{\mathcal Q})
            \end{align*}
    \State Form local regressor
    \[
    \textbf H_\mathcal P(j)
    =
    (\textbf R_\mathcal P(j)^\top\otimes\textbf L_\mathcal P(j))\,\textbf S_\mathcal P
    \]
\EndFor
\State Stack $\textbf H_\mathcal P=[\textbf H_\mathcal P(1)^\top\ \textbf H_\mathcal P(2)^\top \ \cdots\ \textbf H_\mathcal P(T)^\top]^\top$
\State \Return $\textbf H_\mathcal P$
\end{algorithmic}
\vspace{-17pt}
\rule{\linewidth}{0.6pt}

\newpage

\end{document}